\documentclass[a4paper, 11pt,oneside]{article}
\usepackage{amsmath}
\usepackage{amsthm}
\usepackage{amssymb}
\usepackage{bbm}
\usepackage{graphicx}
\usepackage{authblk}
\usepackage[margin=1.2in]{geometry}
\usepackage{tikz}
\usepackage{caption}
\usetikzlibrary{calc,automata,arrows,shapes,positioning}

\theoremstyle{plain}
\newtheorem{theorem}{Theorem}
\newtheorem{lemma}{Lemma}
\newtheorem{corollary}{Corollary}
\newtheorem{proposition}{Proposition}

\newtheorem{property}{Property}

\newtheorem{definition}{Definition}

\theoremstyle{remark}

\newcommand{\swc}[1]{| #1 |_s}
\newcommand{\ssl}{\mbox{\sf ssl}}
\newcommand{\sw}{\mbox{\sf swc}}
\newcommand{\SSL}{\mbox{\sf SSL}}
\newcommand{\SW}{\mbox{\sf SWC}}

\title{Counting symbol switches in synchronizing automata}
\author[1]{Henk Don}
\author[2,3]{Hans Zantema}
\affil[1]{\small Radboud University Nijmegen, P.O.\ Box 9010, 6500 GL Nijmegen, The Netherlands,
	email: {\tt h.don@math.ru.nl}}
\affil[2]{Department of Computer Science, TU Eindhoven, P.O.\ Box 513,
		5600 MB Eindhoven, The Netherlands, 
		email: {\tt h.zantema@tue.nl}}
\affil[3]{Radboud University Nijmegen, P.O.\ Box 9010, 6500 GL Nijmegen, The Netherlands}

\begin{document}
\maketitle

\begin{abstract}
	Instead of looking at the lengths of synchronizing words as in \v{C}ern\'y's conjecture, we look at the {\em switch count} of such words,
	that is, we only count the switches from one letter to another. Where the synchronizing words of the \v{C}ern\'y automata $\mathcal{C}_n$ have switch count linear in $n$, we wonder whether synchronizing automata exist for which every synchronizing word has quadratic switch count. The answer is positive: we prove that switch count has the same complexity as synchronizing word length. We give some series of synchronizing automata yielding quadratic switch count, the best one reaching $\frac{2}{3} n^2 + O(n)$ as switch count.
	
	We investigate all binary automata on at most 9 states and determine the maximal possible switch count. For all $3\leq n\leq 9$, a strictly higher switch count can be reached by allowing more symbols. This behaviour differs from length, where for every $n$, no automata are known with higher synchronization length than $\mathcal{C}_n$, which has only two symbols. It is not clear if this pattern extends to larger $n$. For $n\geq 12$, our best construction only has two symbols.
	
\end{abstract}

%
%
%
\section{Introduction}

The well-known \v{C}ern\'y automaton ${\cal C}_n$ on $n$ states has shortest synchronizing word $b(a^{n-1}b)^{n-2}$ of length $(n-1)^2$; for $n=4$ this is $baaabaaab$ and the automaton is drawn below.

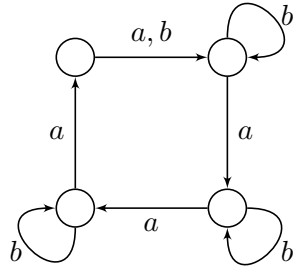
\begin{center}
	\begin{tikzpicture}[-latex',node distance =2 cm and 2cm ,on grid ,
	semithick , state/.style ={ circle ,top color =white , bottom
		color = white!20 , draw, black , text=black , minimum width =.5
		cm}]
	
	\node[state] (1) {};
	\node[state] (4) [below =of 1] {};
	\node[state] (2) [right =of 1] {};
	\node[state] (3) [right =of 4] {};
	\path (1) edge node[above] {$a,b$} (2);
	\path (2) edge node[right] {$a$} (3);
	\path (3) edge node[below] {$a$} (4);
	\path (4) edge node[left] {$a$} (1);
	\path (2) edge [loop right,looseness=8, in = 0, out = 90] node[right] {$b$} (2);
	\path (3) edge [loop right,looseness=8, in = 270, out = 0] node[right] {$b$} (3);
	\path (4) edge [loop right,looseness=8, in = 180, out = 270] node[left] {$b$} (4);
	\end{tikzpicture}
\end{center}

\v{C}ern\'y's conjecture \cite{C64} states that for every synchronizing DFA on $n$ states the shortest synchronizing word has length at most $(n-1)^2$.
Moreover, it is conjectured that apart from a few exceptions for $n \leq 6$, this DFA is the only one for which this length $(n-1)^2$ is reached.
It is clear that the shortest synchronizing word $b(a^{n-1}b)^{n-2}$ has a lot of big groups of consecutive $a$'s, and only a linear number of switches from $a$ to $b$ or conversely.

Now we wonder whether this is always the case: if a DFA has a long shortest synchronizing word, is it always the case that this big length is caused by big number of consecutive equal symbols? The answer will be negative: in this paper we investigate synchronizing DFAs on $n$ states of which shortest synchronizing words both have length quadratic in $n$
and do not contain many consecutive equal symbols. To make the question more precise we introduce the notion of {\em switch count}: an alternative notion of word size ignoring consecutive equal symbols: the switch count $\swc{w}$ of a word $w$ is the length of $w$ after collapsing consecutive equal symbols to a single symbol. So $\swc{b(a^{n-1}b)^{n-2}} = |b(ab)^{n-2}| = 2n-3$. The switch count is inductively defined by
\[ \swc{\epsilon} = 0, \; \swc{a} = 1,\; \swc{aaw} = \swc{aw},\; \swc{abw}= 1+ \swc{b}\]
for all symbols $a \neq b$ and all words $w$.

In this paper we consider minimal switch counts of synchronizing words, rather than minimal lengths of synchronizing words, and we look for automata for which this number is high.
More precisely, we define the {\em switch count} $\sw(A)$ of a synchronizing automaton $A$ to be the minimal switch count of a synchronizing word, and we are looking for synchronizing automata
for which this switch count is as high as possible, in particular, quadratic in the number of states.

An alternative formulation is as follows. Define an automaton to be {\em power closed} if for every symbol $a$ and every $k>1$ the operation $a^k$ is either the identity
or occurs as a symbol in the automaton. The {\em power closure} of any automaton is defined to be the automaton obtained by adding every operation $a^k$ for every symbol $a$ and every $k>1$, as long as this operation is not the identity and does not yet occur as a symbol. Now the switch count of any synchronizing automaton coincides with the minimal synchronizing word length of its power closure. This observation follows from the definitions, since in switch count every consecutive groups of equal symbols is counted as a single symbol.
So the maximal switch count of any synchronizing automaton on $n$ states coincides with the maximal synchronizing word length of any power closed automaton on $n$ states.

While \v{C}ern\'y's conjecture is about maximizing the minimal {\em length} of a synchronizing word, in this paper we are interested in maximizing the switch count instead.
Note that ${\cal C}_n$ is no candidate at all for a quadratic bound on switch counts, as we saw that its switch count $\sw({\cal C}_n)$ is at most $2n-3$.
In fact we have $\sw({\cal C}_n) = 2n-3$ since every synchronizing word of switch count $< 2n-3$ can be transformed to a synchronizing word of length $< (n-1)^2$
by replacing every $a^k$ by $a^{k \mod n}$ and every $b^k$ by $b$, for every $k>1$.

A first main result is that when only looking at the order of magnitude, there is no difference between switch count and synchronization length. More precisely, writing
$\SW(n)$ for the largest switch count of a synchronizing automaton on $n$ states, and $\SSL(n)$ for the largest shortest synchronizing word length of a
synchronizing automaton on $n$ states, we have $\SW(n) = \Theta(\SSL(n))$.

In our analysis we focus on two flavors of automata: binary automata having only two symbols and automata in which any number of symbols is allowed. When looking at synchronization length, having more symbols does not give more power:
the highest known synchronization length $(n-1)^2$ for automata on $n$ states is achieved by the binary \v{C}ern\'y automaton for every $n > 1$. For an experimental analysis of automata on restricted number of states, we refer to \cite{T06,BDZ17}.

For switch count for small $n$ this is substantially different, as is shown in the following table giving maximal switch counts for automata on $n \leq 12$ states:
\[ \begin{array}{|c|c|c|}
\mbox{nr of states $n$} & \mbox{binary} & \mbox{any nr of symbols} \\
\hline
2 & 1 & 1 \\
\hline
3 & 3 & 4 \\
\hline
4 & 7 & 9 \\
\hline
5 & 11 & 15 \\
\hline
6 & 19 & \geq 21 \\
\hline
7 & 25 & \geq 28 \\
\hline
8 & 31 & \geq 36 \\
\hline
9 & 41 & \geq 45 \\
\hline
10 & \geq 53 & \geq 55 \\
\hline
11 & \geq 65 & \geq 66 \\
\hline
12 & \geq 79 & \geq 79 \\
\hline
\end{array} \]
Where this table gives an exact number this is found by exhaustive search on all automata of the indicated shape. Where it states '$\geq$' it is the highest switch count of a particular automaton that we could find. So for all $n$ with $3 \leq n \leq 9$ we are sure that non-binary allows for a higher switch count than binary.

Most of the particular automata yielding the results in the table are instances from series of automata parametrized by the number of states $n$. For unbounded number of symbols and $5 \leq n \leq 11$ states
the bound is reached by the automaton ${\cal R}_n$ introduced in Section \ref{secbc}, for which we prove that its switch count is $\frac{n(n+1)}{2}$. For $n=3,4$ the automata
T3-2 and T4-1 as given in \cite{DZ16} yield the given switch counts $4,9$.

For binary automata and $n \geq 8$ states the bound is reached by the automaton ${\cal A}_n$ presented and investigated in Section \ref{secmr}. For lower values of $n$ we give examples in Section \ref{secexp}. To give the flavor of ${\cal A}_n$ we show ${\cal A}_7$.

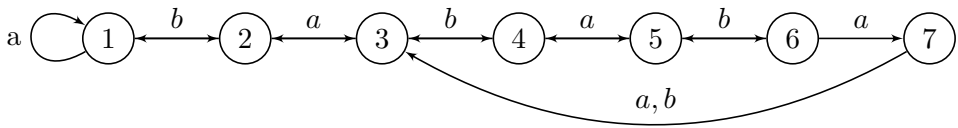
\begin{center}
	\begin{tikzpicture}[-latex',node distance =1.8 cm and 1.8cm ,on grid ,
	semithick , state/.style ={ circle ,top color =white , bottom
		color = white!20 , draw, black , text=black , minimum width =.5
		cm}]
	
	\node[state] (1) {1};
	\node[state] (2) [right =of 1] {2};
	\node[state] (3) [right =of 2] {3};
	\node[state] (4) [right =of 3] {4};
	\node[state] (5) [right =of 4] {5};
	\node[state] (6) [right =of 5] {6};
	\node[state] (7) [right =of 6] {7};
	\path (1) edge [bend right = 0] node[above] {$b$} (2);
	\path (2) edge [bend right = 0] node[above] {$a$} (3);
	\path (3) edge [bend right = 0] node[above] {$b$} (4);
	\path (4) edge [bend right = 0] node[above] {$a$} (5);
	\path (5) edge [bend right = 0] node[above] {$b$} (6);
	\path (6) edge [bend right = 0] node[above] {$a$} (7);
	
	\path (2) edge [bend right = 0] node[above] {} (1);
	\path (3) edge [bend right = 0] node[above] {} (2);
	\path (4) edge [bend right = 0] node[above] {} (3);
	\path (5) edge [bend right = 0] node[above] {} (4);
	\path (6) edge [bend right = 0] node[above] {} (5);
	
	\path (7) edge [bend left = 30] node[above] {$a,b$} (3);
	
	\path (1) edge [loop left,looseness=8, in = 150, out = 210] node[left] {a} (1);
	\end{tikzpicture}
\end{center}

More general, in ${\cal A}_n$ the states are numbered from $1$ to $n$, and the two symbols $a,b$ satisfy
\begin{itemize}
	\item $a$ swaps $2k+1$ and $2k$, and $b$ swaps $2k-1$ and $2k$, for $k > 0$,
	\item $1a = 1$, $na = nb = q$ for $q \approx n/3$;
\end{itemize}
details are given in Section \ref{secmr}. The most involved result of this paper given in Section \ref{secmr} proves that ${\cal A}_n$
has switch count $\left\lceil \frac{2}{3}n(n-2)-1\right\rceil$. For $n \geq 8$ it is the highest switch of any binary automaton that we know, and for $n \geq 12$ it is the highest switch count
of any automaton that we know. A remarkable observation is that the reasons for having high switch counts are completely different for ${\cal R}_n$ and ${\cal A}_n$.
In ${\cal R}_n$ this is because many symbols swap two states and are the identity on all other states, and for synchronization a kind of bubble sort on the states is required. The maximal number of steps needed to synchronize a pair of states in $\mathcal{R}_n$ is linear in $n$. In contrast, in ${\cal A}_n$ both symbols swap nearly all the states. As a consequence, most time only strings of the shape $(ab)^k$ are relevant.  They act in such a way that the maximal number of steps needed to synchronize a pair is quadratic in $n$.

This is paper is organized as follows. In Section \ref{secprel} we give some preliminaries.
In Section \ref{seccomp} we give a transformation on synchronizing automata for which the switch count of the transformed automaton is of the same order as the synchronization length of the original one, from which $\SW(n) = \Theta(\SSL(n))$ is concluded.
In Section \ref{secexp} we present some experimental results on binary automata, mainly consisting of examples reaching the highest possible switch count for given values of $n$.
In Section \ref{secbc} we give some basic constructions for automata having quadratic switch count, all of the shape $\frac{1}{2} n^2 + O(n)$, including ${\cal R}_n$ yielding the
highest known switch counts for $5 \leq n \leq 11$ . In Section \ref{secmr} we give our main result: the analysis of the automata ${\cal A}_n$ improving the switch count from $\frac{1}{2} n^2 + O(n)$ to $\frac{2}{3} n^2 + O(n)$. In Section \ref{seccyc} we show that for cyclic automata, that is, one symbol cyclicly permutes all states, the switch count is linear in the number $n$ of states in case $n$ is a prime number. We conclude in Section \ref{secconcl}.

\section{Preliminaries}
\label{secprel}

A {\em deterministic finite automaton (DFA)} over a finite alphabet $\Sigma$ consists of a finite set
$Q$ of states and a map $\delta: Q \times \Sigma \to Q$.\footnote{For synchronization the
	initial state and the set of final states in the standard definition may be ignored.}
For $w \in \Sigma^*$ and $q \in Q$ define $qw$ inductively by
$q \epsilon = q$ and $q w a = \delta(qw,a)$ for
$a \in \Sigma$. So  $qw$ is the state where one ends when starting in $q$ and applying
$\delta$-steps for the symbols in $w$ consecutively, and $qa$ is a short hand notation for $\delta(q,a)$.
A word $w \in \Sigma^*$ is called {\em synchronizing} if a
state $q_s \in Q$ exists such that $q w = q_s$ for all $q \in Q$. Stated in words: starting in
any state $q$, after processing $w$ one always ends in state $q_s$.
Obviously, if $w$ is a synchronizing word then so is $wu$ for any word $u$.

The basic tool to analyze synchronization is by exploiting the {\em power automaton}.
For any DFA $(Q,\Sigma, \delta)$ its power automaton is the DFA $(2^Q,\Sigma, \delta')$
where $\delta' : 2^Q \times \Sigma \to 2^Q$ is defined by
$\delta'(V,a) = \{q \in Q \mid \exists p \in V : \delta(p,a) = q \}$.
For any $V \subseteq Q, w \in \Sigma^*$ we define $Vw$ as above, using $\delta'$ instead of
$\delta$.  From this definition one easily proves that
$Vw = \{ qw \mid q \in V \}$ for any $V \subseteq Q, w \in \Sigma^*$.
A set of the shape $\{q\}$ for $q \in Q$ is called a {\em singleton}.
So a word $w$ is synchronizing if and only if $Qw$ is a singleton.
Hence a DFA is synchronizing if and only if its
power automaton admits a path from $Q$ to a singleton, and the shortest
length of such a path corresponds to the shortest length of a synchronizing word.

Similarly, the shortest length of a path from $Q$ to a singleton in the power automaton of the
power closure corresponds to the switch count of the automaton. In the power automaton of the original automaton
this switch count corresponds to the switch count of a word corresponding to such a path.

For $V \subseteq Q$ and $a \in \Sigma$ we define $Va^{-1} = \{q \in Q \mid qa \in V\}$, and for a single state $q$ we
shortly write $q a^{-1}$ instead of $\{q\} a^{-1}$. In the power automaton we can also reason backward: a word $w = a_1 a_2 \cdots a_k$
is synchronizing if and only if a state $q$ exists such that $q a_k^{-1} a_{k-1}^{-1} \cdots a_2^{-1} a_1^{-1} = Q$.

For the concatenation $w_1 w_2 \cdots w_k$ of words $w_i \in \Sigma^*$ we shortly write $\prod_{i=1}^k w_i$.

\section{Complexity of length and switch count}
\label{seccomp}

For any synchronizing automaton $A$, let $\ssl(A)$ denote the shortest synchronizing word length of $A$.
By definition $\swc{w} \leq |w|$ for any word $w$, so $\sw(A) \leq \ssl(A)$ for any synchronizing automaton $A$. We already saw that $\sw({\cal C}_n) = 2n-3$ while $\ssl({\cal C}_n) = (n-1)^2$, so $\sw(A)$ may be much smaller than $\ssl(A)$. In this section we will show that nevertheless the highest switch count of any synchronizing automaton on $n$ states is of the same order of magnitude as the highest synchronizing word length of any synchronizing automaton on $n$ states. The key idea is to find a transformation of any synchronizing automaton $A$ on $n$ states to a synchronizing automaton $F(A)$ on $2n$ states such that $\sw(F(A)) = 2 \ssl(A)$.

For a synchronizing automaton $A = (Q,\Sigma,\delta)$ we define the automaton $F(A) = (Q \cup Q', \Sigma \cup \{c\}, \delta')$, where $Q'$ is a disjoint copy of Q, for which we write
$q' \in Q'$ for the element corresponding to $q \in Q$, $c$ is a fresh symbol not in $\Sigma$, and $\delta' : (Q \cup Q') \times (\Sigma \cup \{c\}) \to  Q \cup Q'$ is defined by
\[ \delta'(q,a) = q, \; \delta'(q,c) = q',\; \delta'(q',a) = \delta(q,a), \; \delta'(q',c) = q',\]
for all $q \in Q, a \in \Sigma$. For instance, if $\Sigma = \{a,b\}$, then by applying $F$ the pattern

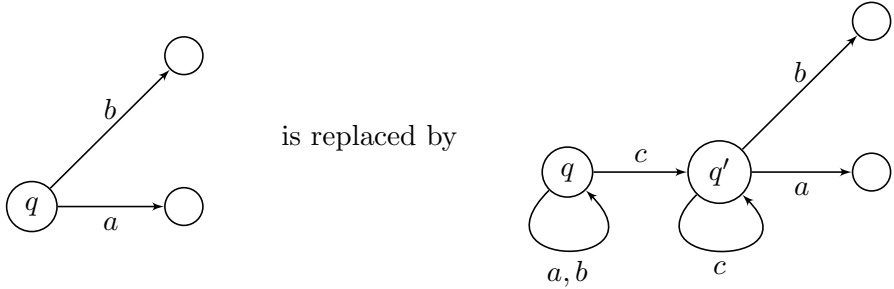
\begin{minipage}{35mm}
	\begin{tikzpicture}[-latex',node distance =2 cm and 2cm ,on grid ,
	semithick , state/.style ={ circle ,top color =white , bottom
		color = white!20 , draw, black , text=black , minimum width =.5
		cm}]
	
	\node[state] (1) {$q$};
	\node[state] (3) [right =of 1] {};
	\node[state] (4) [above =of 3] {};
	\path (1) edge node[above] {$b$} (4);
	\path (1) edge node[below] {$a$} (3);
	\end{tikzpicture}
\end{minipage}
is replaced by
\begin{minipage}{6cm}
	\vspace{3mm}
	
	\begin{tikzpicture}[-latex',node distance =2 cm and 2cm ,on grid ,
	semithick , state/.style ={ circle ,top color =white , bottom
		color = white!20 , draw, black , text=black , minimum width =.5
		cm}]
	
	\node[state] (1) {$q$};
	\node[state] (2) [right =of 1] {$q'$};
	\node[state] (3) [right =of 2] {};
	\node[state] (4) [above =of 3] {};
	\path (1) edge node[above] {$c$} (2);
	\path (2) edge node[above] {$b$} (4);
	\path (2) edge node[below] {$a$} (3);
	\path (1) edge [loop right,looseness=8, in = 315, out = 225] node[below] {$a,b$} (1);
	\path (2) edge [loop right,looseness=6, in = 315, out = 225] node[below] {$c$} (2);
	\end{tikzpicture}
\end{minipage}

\begin{theorem}
	\label{thmcomp}
	Let $A$ be a synchronizing automaton. Then $F(A)$ is synchronizing too, and $\sw(F(A)) = 2 \ssl(A)$.
\end{theorem}
\begin{proof}
	Define the homomorphism $h : \Sigma \to (\Sigma \cup \{c\})^*$ by $h(a) = ca$ for all $a \in \Sigma$.
	Then in $F(A)$ one has $qca = q'ca = \delta(q,a)$ for all $q \in Q, a\in \Sigma$. So if $w$ is a synchronizing word for $A$, then $h(w)$ is a synchronizing word for
	$F(A)$, proving that $F(A)$ is synchronizing.
	
	Since $qc^k = qc$ and $qa^k = qa$ for all $q \in Q \cup Q', a \in \Sigma, k >0$, we conclude that $F(A)$ is power closed. Hence $\sw(F(A)) = \ssl(F(A))$.
	So it remains to prove that $\ssl(F(A)) = 2 \ssl(A)$. Let $w$ be a synchronizing word for $A$ of length $\ssl(A)$, then $h(w)$ is a synchronizing word for
	$F(A)$ of length $2\ssl(A)$, proving that $\ssl(F(A)) \leq 2 \ssl(A)$. Conversely, let $v$ be a shortest synchronizing word for $F(A)$. Using 'shortest', we observe
	\begin{itemize}
		\item $(Q \cup Q')ac = Q' = (Q \cup Q')c$, so the first element of $v$ is $c$.
		\item $qau = qa$ for all $a \in \Sigma, u \in \Sigma^*, q \in Q \cup Q'$, so no two consecutive elements in $v$ are in $\Sigma$.
		\item $qc^k = qc$ for all $a \in \Sigma, k > 0, q \in Q \cup Q'$, so no two consecutive $c$'s occur in $v$.
		\item If $v'c$ is synchronizing for $F(A)$, then so is $v'$, so the last element of $v$ is in $\Sigma$.
	\end{itemize}
	For these observations we conclude that $v = h(w)$ for some $w \in \Sigma^*$. Since $ca$ acts on $Q$ in $F(A)$ just like $a$ on $Q$ in $A$, for every $a \in \Sigma$,
	we conclude that $w$ is a synchronizing word for $A$, hence $|v| = |h(w)| = 2|w| \geq 2 \ssl(A)$, concluding the proof.
\end{proof}

The picture below shows $F({\cal C}_4)$.

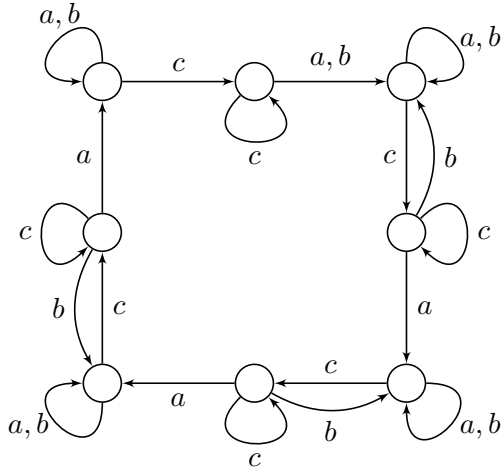
\begin{center}
	\begin{tikzpicture}[-latex',node distance =2 cm and 2cm ,on grid ,
	semithick , state/.style ={ circle ,top color =white , bottom
		color = white!20 , draw, black , text=black , minimum width =.5
		cm}]
	
	\node[state] (1) {};
	\node[state] (8) [below =of 1] {};
	\node[state] (4) [below =of 8] {};
	\node[state] (5) [right =of 1] {};
	\node[state] (2) [right =of 5] {};
	\node[state] (6) [below =of 2] {};
	\node[state] (7) [right =of 4] {};
	\node[state] (3) [right =of 7] {};
	\path (1) edge node[above] {$c$} (5);
	\path (5) edge node[above] {$a,b$} (2);
	\path (2) edge node[left] {$c$} (6);
	\path (6) edge node[right] {$a$} (3);
	\path (3) edge node[above] {$c$} (7);
	\path (7) edge node[below] {$a$} (4);
	\path (4) edge node[right] {$c$} (8);
	\path (8) edge node[left] {$a$} (1);
	\path (6) edge [looseness=1, in = 300, out = 60] node[right] {$b$} (2);
	\path (7) edge [looseness=1, in = 210, out = 330] node[below] {$b$} (3);
	\path (8) edge [looseness=1, in = 120, out = 240] node[left] {$b$} (4);
	\path (5) edge [loop right,looseness=8, in = 315, out = 225] node[below] {$c$} (5);
	\path (6) edge [loop right,looseness=8, in = 315, out = 45] node[right] {$c$} (6);
	\path (7) edge [loop right,looseness=8, in = 315, out = 225] node[below] {$c$} (7);
	\path (8) edge [loop right,looseness=8, in = 225, out = 135] node[left] {$c$} (8);
	\path (1) edge [loop right,looseness=8, in = 180, out = 90] node[above] {$a,b$} (1);
	\path (2) edge [loop right,looseness=8, in = 0, out = 90] node[right] {$a,b$} (2);
	\path (3) edge [loop right,looseness=8, in = 270, out = 0] node[right] {$a,b$} (3);
	\path (4) edge [loop right,looseness=8, in = 180, out = 270] node[left] {$a,b$} (4);
	\end{tikzpicture}
\end{center}

From Theorem \ref{thmcomp} we conclude that the automaton $F({\cal C}_n)$ on $2n$ states has switch count $2(n-1)^2$ for every $n > 1$. So for every even number $n$ we have a synchronizing  automaton on $n$ states with switch count $\frac{n^2 -4n+4}{2}$. In particular, the automaton on 8 states sketched above has switch count 18.

The following key theorem is an easy consequence of Theorem \ref{thmcomp}.

\begin{theorem}
	The functions $\SW(n)$ and $\SSL(n)$ satisfy
	\[\SW(n) = \Theta(\SSL(n)).\]
\end{theorem}
\begin{proof}
	By definition we have $\SW(n) \leq \SSL(n)$, so $\SW(n) = O(\SSL(n))$. From Theorem \ref{thmcomp} we conclude $\SW(2n) \geq 2 \SSL(n)$. Since $\SW(n)$ is polynomial in $n$,
	we have $\SW(n) \geq C\cdot \SW(2n)$ for some $C > 0$, combining to $\SW(n) = \Omega(\SSL(n))$.
\end{proof}

A similar result holds when restricting to binary automata. For the proof the construction $F$ is adjusted to $F_2$ as follows: suppose $\Sigma = \left\{a,b\right\}$, replace $c$ by $ab$ and introduce extra states $q''$ in the middle of each $c$ step. More precisely,
$F_2(A) = (Q \cup Q' \cup Q'', \{a,b\}, \delta')$, where $Q',Q''$ are two disjoint copies of Q, and $\delta' : (Q \cup Q' \cup Q'') \times \{a,b\} \to  Q \cup Q' \cup Q''$ is defined by
$\delta'(q,a) = q''$, $\delta'(q,b) = q$, $\delta'(q',a) = \delta(q,a)$, $\delta'(q',b) = \delta(q,b)$, $\delta'(q'',a) = q''$, $\delta'(q'',b) = q'$,
for all $q \in Q$. It can be shown that $\sw(F_2(A)) = 2\ssl(A)$ for any synchronizing automaton $A$ by an argument similar to the proof of Theorem \ref{thmcomp}, although $\sw(F_2(A))$ is not power closed.

\section{Experimental results for binary automata}
\label{secexp}

In this section we present some examples of binary automata reaching the maximal possible switch count for given number of states. For $n\leq 9$, we did an exhaustive search, leading to the maximal switch counts given in the table below. We also identified how many non-isomorfic binary automata reach this maximum.
\[
\begin{array}{c||c|c|c|c|c|c|c}
n & 3 & 4 & 5 & 6 & 7 & 8 & 9\\
\hline
\text{maximum} &\ \ 3\ \ &\ \ 7\ \ &\ \ 11\ \ &\ \ 19\ \ &\ \ 25\ \ &\ \ 31\ \ &\ \ 41\ \ \\
\text{\# automata}& 6 & 2 & 6 & 2 & 2 & 8 & 2
\end{array}
\]
For each value of $n$, we will give an example attaining the maximal possible value. For $n=8$ and $n=9$, the maximum is already attained by ${\cal A}_n$, to be investigated in Section \ref{secmr}. Here we will also give the other extremal automaton for $n=9$, and for $n=8$ we will give one more example as well. As was to be expected, the extremal automata typically only have short cycles.

%
%
%
%
%
%
%
%
%
%
%
%
%
%
%
%

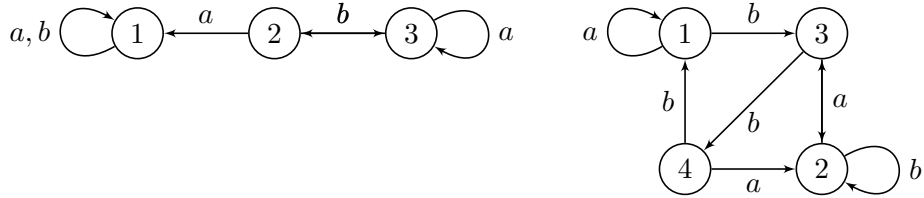
\begin{center}
	\begin{tikzpicture}[-latex',node distance =1.8 cm and 1.8cm ,on grid ,
	semithick , state/.style ={ circle ,top color =white , bottom
		color = white!20 , draw, black , text=black , minimum width =.5
		cm}]
	
	\node[state] (1) {1};
	\node[state] (2) [right=of 1] {2};
	\node[state] (3) [right =of 2] {3};
	
	\path (2) edge [bend right = 0] node[above] {$a$} (1);
	
	\path (1) edge [loop left,looseness=8, in = 150, out = 210] node[left] {$a,b$} (1);
	\path (3) edge [loop left,looseness=8, in = -30, out = 30] node[right] {$a$} (3);

	\path (2) edge [bend right = 0] node[above] {$b$} (3);
	\path (3) edge [bend right = 0] node[above] {$b$} (2);
	
	

	%
	%
	%

	\node[state,white] (0) [right = of 3] {}; 
	
	\node[state] (11) [right =of 0] {1};
	\node[state] (13) [right =of 11] {3};
	\node[state] (12) [below =of 13] {2};
	\node[state] (14) [below =of 11] {4};
	
	
	\path (12) edge [bend right = 0] node[right] {$a$} (13);
	\path (13) edge [bend right = 0] node[above] {} (12);
	\path (14) edge [bend right = 0] node[below] {$a$} (12);
	\path (11) edge [loop left,looseness=8, in = 150, out = 210] node[left] {$a$} (11);

	\path (11) edge [bend right = 0] node[above] {$b$} (13);
	\path (13) edge [bend right = 0] node[below] {$b$} (14);
	\path (14) edge [bend right = 0] node[left] {$b$} (11);
	
	\path (12) edge [loop left,looseness=8, in = -30, out = 30] node[right] {$b$} (12);
	
	\end{tikzpicture}
\end{center}

The automaton on three states above has switch count 3. It has a unique shortest synchronizing word $w=aba$, with $|w|=|w|_s=3$. Note that there never is a unique synchronizing word with minimal switch count, since there are always trivial extensions at the beginning and the end. Also at other positions often letters can be added. In this example on three states, every word of the form $ab^ka$ with $k$ odd is synchronizing as well.

In the examples that follow, we will always give a word with minimal switch count which is as short as possible. Under this restriction, we can again ask if this word is unique. In particular, the answer is positive if there is a unique shortest synchronizing word which happens to have minimal switch count. Remarkably, it is possible that the shortest synchronizing word does not have minimal switch count, as we will see in our example with 8 states.

The automaton on four states has switch count 7. Also this one has a unique shortest synchronizing word $w=abab^2aba$, with $|w|=8$ and $|w|_s=7$.

%
%
%
%
%
%
%
%
%
%
%
%

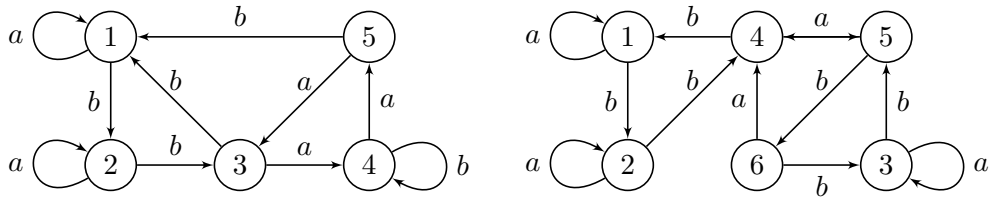
\begin{center}
	\begin{tikzpicture}[-latex',node distance =1.7 cm and 1.7cm ,on grid ,
	semithick , state/.style ={ circle ,top color =white , bottom
		color = white!20 , draw, black , text=black , minimum width =.5
		cm}]
	
	\node[state] (1) {1};
	\node[state] (2) [below =of 1] {2};
	\node[state] (3) [right =of 2] {3};
	\node[state] (4) [right =of 3] {4};
	\node[state] (5) [above =of 4] {5};
	
	\path (3) edge [bend right = 0] node[above] {$a$} (4);
	\path (4) edge [bend right = 0] node[right] {$a$} (5);
	\path (5) edge [bend right = 0] node[above] {$a$} (3);
	
	\path (1) edge [loop left,looseness=8, in = 150, out = 210] node[left] {$a$} (1);
	\path (2) edge [loop left,looseness=8, in = 150, out = 210] node[left] {$a$} (2);

	\path (1) edge [bend right = 0] node[left] {$b$} (2);
	\path (2) edge [bend right = 0] node[above] {$b$} (3);
	\path (3) edge [bend right = 0] node[above] {$b$} (1);
	\path (5) edge [bend right = 0] node[above] {$b$} (1);
	
	\path (4) edge [loop left,looseness=8, in = -30, out = 30] node[right] {$b$} (4);
	

	%
	%
	%

	\node[state,white] (0) [right = of 5] {}; 
	
	\node[state] (11) [right =of 0] {1};
	\node[state] (12) [below =of 11] {2};
	\node[state] (14) [right =of 11] {4};
	\node[state] (16) [below =of 14] {6};
	\node[state] (15) [right =of 14] {5};
	\node[state] (13) [below =of 15] {3};
	
	
	\path (14) edge [bend right = 0] node[right] {} (15);
	\path (15) edge [bend right = 0] node[above] {$a$} (14);
	\path (16) edge [bend right = 0] node[left] {$a$} (14);
	\path (11) edge [loop left,looseness=8, in = 150, out = 210] node[left] {$a$} (11);
	\path (12) edge [loop left,looseness=8, in = 150, out = 210] node[left] {$a$} (12);
	\path (13) edge [loop left,looseness=8, in = -30, out = 30] node[right] {$a$} (13);

	\path (11) edge [bend right = 0] node[left] {$b$} (12);
	\path (12) edge [bend right = 0] node[above] {$b$} (14);
	\path (13) edge [bend right = 0] node[right] {$b$} (15);
	\path (14) edge [bend right = 0] node[above] {$b$} (11);
	\path (15) edge [bend right = 0] node[above] {$b$} (16);
	\path (16) edge [bend right = 0] node[below] {$b$} (13);
	
	
	\end{tikzpicture}
\end{center}

On five states we have an automaton with switch count 11. It has a unique shortest synchronizing word $w=ba^2baba^2bab^2a^2b$ with $|w|=15$ and $|w|_s=11$. In this case, there exist synchronizing words with the same switch count that can not be obtained by merely adding some letters to $w$, for instance $ba^2bab^2a^2b^2aba^2b$ of length 16. The example on six states has switch count 19 and unique shortest synchronizing word $w = abab^2ababab^2ababa^2b^2aba$, satisfying $|w|=23$ and $|w|_s=19$.

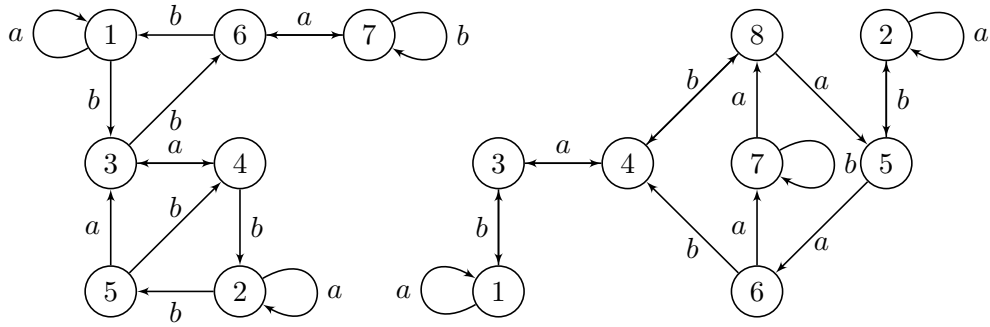
\begin{center}
	\begin{tikzpicture}[-latex',node distance =1.7 cm and 1.7cm ,on grid ,
	semithick , state/.style ={ circle ,top color =white , bottom
		color = white!20 , draw, black , text=black , minimum width =.5
		cm}]
	
	\node[state] (1) {1};
	\node[state] (3) [below =of 1] {3};
	\node[state] (4) [right =of 3] {4};
	\node[state] (2) [below =of 4] {2};
	\node[state] (5) [below =of 3] {5};
	\node[state] (6) [right =of 1] {6};
	\node[state] (7) [right =of 6] {7};
	
	
	\path (3) edge [bend right = 0] node[above] {$a$} (4);
	\path (4) edge [bend right = 0] node[right] {} (3);
	\path (5) edge [bend right = 0] node[left] {$a$} (3);
	\path (6) edge [bend right = 0] node[above] {$a$} (7);
	\path (7) edge [bend right = 0] node[below] {} (6);
	\path (1) edge [loop left,looseness=8, in = 150, out = 210] node[left] {$a$} (1);
	\path (2) edge [loop left,looseness=8, in = -30, out = 30] node[right] {$a$} (2);

	\path (1) edge [bend right = 0] node[left] {$b$} (3);
	\path (2) edge [bend right = 0] node[below] {$b$} (5);
	\path (3) edge [bend right = 0] node[below] {$b$} (6);
	\path (4) edge [bend right = 0] node[right] {$b$} (2);
	\path (5) edge [bend right = 0] node[above] {$b$} (4);
	\path (6) edge [bend right = 0] node[above] {$b$} (1);
	
	\path (7) edge [loop below,looseness=8, in = -30, out = 30] node[right] {$b$} (7);
	
	
	\node[state,white] (0) [right = of 7] {}; 
	\node[state] (14) [below right = of 0] {4};
	\node[state] (13) [left =of 14] {3};
	\node[state] (11) [below =of 13] {1};
	\node[state] (17) [right =of 14] {7};
	\node[state] (15) [right =of 17] {5};
	\node[state] (16) [below =of 17] {6};
	\node[state] (12) [above =of 15] {2};
	\node[state] (18) [above =of 17] {8};
	
	
	\path (13) edge [bend right = 0] node[above] {$a$} (14);
	\path (14) edge [bend right = 0] node[right] {} (13);
	\path (15) edge [bend right = 0] node[below] {$a$} (16);
	\path (16) edge [bend right = 0] node[left] {$a$} (17);
	\path (17) edge [bend right = 0] node[left] {$a$} (18);
	\path (18) edge [bend right = 0] node[above] {$a$} (15);
	\path (11) edge [loop left,looseness=8, in = 150, out = 210] node[left] {$a$} (11);
	\path (12) edge [loop left,looseness=8, in = -30, out = 30] node[right] {$a$} (12);

	\path (11) edge [bend right = 0] node[left] {$b$} (13);
	\path (12) edge [bend right = 0] node[right] {$b$} (15);
	\path (13) edge [bend right = 0] node[below] {} (11);
	\path (14) edge [bend right = 0] node[above] {$b$} (18);
	\path (15) edge [bend right = 0] node[above] {} (12);
	\path (16) edge [bend right = 0] node[below] {$b$} (14);
	\path (18) edge [bend right = 0] node[above] {} (14);
	
	\path (17) edge [loop left,looseness=8, in = -30, out = 30] node[right] {$b$} (17);

	\end{tikzpicture}
\end{center}
For $n=7$ and $n=8$, the above two automata reach the maximal possible switch count, 25 and 31 respectively. The one on seven states has exactly three shortest synchronizing words of length 32, given by the regular expression
\[
ab^2abab(ab^2a(bab+aba)+bab^2aba)bab^2abab^2ab^2abab^2a.
\]
Al these three words have switch count 25.

The automaton on eight states has the unique shortest synchronizing word
\[
w = ba^3(ba)^3a(ba)^4(ab)^2(ba)^3(ab)^2(ba)^2a(ab)^2
\]
with $|w|=42$ and $|w|_s = 33$. So this word does not have the minimal possible switch count. Instead the minimal switch count is attained by the synchronizing word
\[
w' = ba^3(ba)^3a(ba)^4(ab)^2(ba)^2a^2ba^3ba^2ba^2(ab)^2,
\]
satisfying $|w'|=43>|w|$ and $|w'|_s = 31<|w|_s$.


On $n=9, 10, 11$ states, we have three automata that look very much alike and attain the same switch counts as ${\cal A}_n$. However, the obvious extension to $n=12$ only reaches switch count 77, while ${\cal A}_{12}$ has switch count 79. Below these three automata are given. They all have a unique shortest synchronizing word.

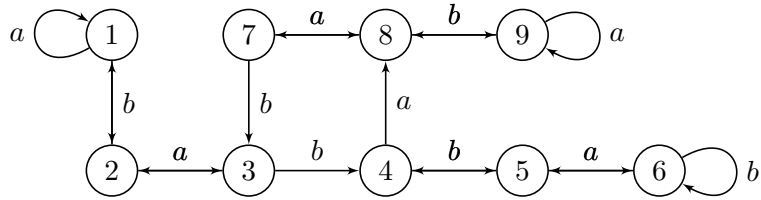
\begin{center}
	\begin{tikzpicture}[-latex',node distance =1.8 cm and 1.8cm ,on grid ,
	semithick , state/.style ={ circle ,top color =white , bottom
		color = white!20 , draw, black , text=black , minimum width =.5
		cm}]
	
	\node[state] (1) {1};
	\node[state] (2) [below =of 1] {2};
	\node[state] (3) [right =of 2] {3};
	\node[state] (4) [right =of 3] {4};
	\node[state] (5) [right =of 4] {5};
	\node[state] (6) [right =of 5] {6};
	\node[state] (7) [right =of 1] {7};
	\node[state] (8) [right =of 7] {8};
	\node[state] (9) [right =of 8] {9};
	
	\path (2) edge [bend right = 0] node[above] {$a$} (3);
	\path (3) edge [bend right = 0] node[above] {$a$} (2);
	\path (4) edge [bend right = 0] node[right] {$a$} (8);
	\path (5) edge [bend right = 0] node[above] {$a$} (6);
	\path (6) edge [bend right = 0] node[above] {$a$} (5);
	\path (7) edge [bend right = 0] node[above] {$a$} (8);
	\path (8) edge [bend right = 0] node[above] {$a$} (7);
	\path (1) edge [loop left,looseness=8, in = 150, out = 210] node[left] {$a$} (1);
	\path (9) edge [loop left,looseness=8, in = -30, out = 30] node[right] {$a$} (9);
	
	\path (1) edge [bend right = 0] node[right] {$b$} (2);
	\path (2) edge [bend right = 0] (1);
	\path (3) edge [bend right = 0] node[above] {$b$} (4);
	\path (4) edge [bend right = 0] node[above] {$b$} (5);
	\path (5) edge [bend right = 0] node[above] {$b$} (4);
	\path (7) edge [bend right = 0] node[right] {$b$} (3);
	\path (8) edge [bend right = 0] node[above] {$b$} (9);
	\path (9) edge [bend right = 0] node[above] {$b$} (8);
	
	\path (6) edge [loop left,looseness=8, in = -30, out = 30] node[right] {$b$} (6);

	\end{tikzpicture}
\end{center}
Automaton on 9 states with switch count 41. The shortest synchronizing word has length 49 and switch count 41:
\[
b^2(ab)^4b(ab)^5b(ab)^5b(ba)^3(ab)^2b^2(ab)^2.
\]

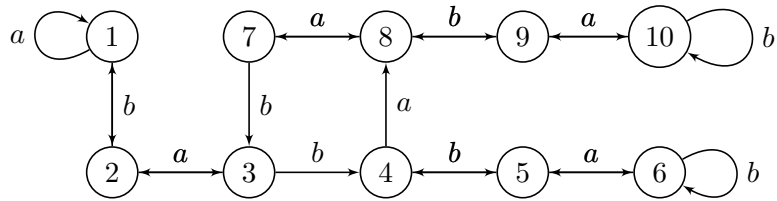
\begin{center}
	\begin{tikzpicture}[-latex',node distance =1.8 cm and 1.8cm ,on grid ,
	semithick , state/.style ={ circle ,top color =white , bottom
		color = white!20 , draw, black , text=black , minimum width =.5
		cm}]
	
	\node[state] (1) {1};
	\node[state] (2) [below =of 1] {2};
	\node[state] (3) [right =of 2] {3};
	\node[state] (4) [right =of 3] {4};
	\node[state] (5) [right =of 4] {5};
	\node[state] (6) [right =of 5] {6};
	\node[state] (7) [right =of 1] {7};
	\node[state] (8) [right =of 7] {8};
	\node[state] (9) [right =of 8] {9};
	\node[state] (10) [right =of 9] {10};
	
	\path (2) edge [bend right = 0] node[above] {$a$} (3);
	\path (3) edge [bend right = 0] node[above] {$a$} (2);
	\path (4) edge [bend right = 0] node[right] {$a$} (8);
	\path (5) edge [bend right = 0] node[above] {$a$} (6);
	\path (6) edge [bend right = 0] node[above] {$a$} (5);
	\path (7) edge [bend right = 0] node[above] {$a$} (8);
	\path (8) edge [bend right = 0] node[above] {$a$} (7);
	\path (9) edge [bend right = 0] node[above] {$a$} (10);
	\path (10) edge [bend right = 0] node[above] {$a$} (9);
	\path (1) edge [loop left,looseness=8, in = 150, out = 210] node[left] {$a$} (1);
	
	\path (1) edge [bend right = 0] node[right] {$b$} (2);
	\path (2) edge [bend right = 0] (1);
	\path (3) edge [bend right = 0] node[above] {$b$} (4);
	\path (4) edge [bend right = 0] node[above] {$b$} (5);
	\path (5) edge [bend right = 0] node[above] {$b$} (4);
	\path (7) edge [bend right = 0] node[right] {$b$} (3);
	\path (8) edge [bend right = 0] node[above] {$b$} (9);
	\path (9) edge [bend right = 0] node[above] {$b$} (8);
	
	\path (6) edge [loop left,looseness=8, in = -30, out = 30] node[right] {$b$} (6);
	\path (10) edge [loop left,looseness=8, in = -30, out = 30] node[right] {$b$} (10);
	
	\end{tikzpicture}
\end{center}
Automaton on 10 states with switch count 53. The shortest synchronizing word has length 63 and switch count 53:
\[
b^2(ab)^4b(ab)^5b(ab)^6(ba)^3(ab)^3b(ba)^3(ab)^2b^2(ab)^2.
\]

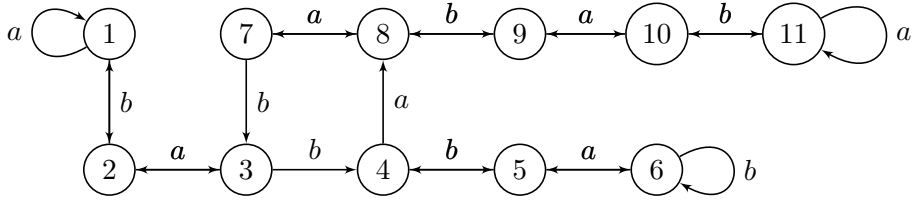
\begin{center}
	\begin{tikzpicture}[-latex',node distance =1.8 cm and 1.8cm ,on grid ,
	semithick , state/.style ={ circle ,top color =white , bottom
		color = white!20 , draw, black, text=black , minimum width =.5
		cm}]
	
	\node[state] (1) {1};
	\node[state] (2) [below =of 1] {2};
	\node[state] (3) [right =of 2] {3};
	\node[state] (4) [right =of 3] {4};
	\node[state] (5) [right =of 4] {5};
	\node[state] (6) [right =of 5] {6};
	\node[state] (7) [right =of 1] {7};
	\node[state] (8) [right =of 7] {8};
	\node[state] (9) [right =of 8] {9};
	\node[state] (10) [right =of 9] {10};
	\node[state] (11) [right =of 10] {11};
	
	\path (2) edge [bend right = 0] node[above] {$a$} (3);
	\path (3) edge [bend right = 0] node[above] {$a$} (2);
	\path (4) edge [bend right = 0] node[right] {$a$} (8);
	\path (5) edge [bend right = 0] node[above] {$a$} (6);
	\path (6) edge [bend right = 0] node[above] {$a$} (5);
	\path (7) edge [bend right = 0] node[above] {$a$} (8);
	\path (8) edge [bend right = 0] node[above] {$a$} (7);
	\path (9) edge [bend right = 0] node[above] {$a$} (10);
	\path (10) edge [bend right = 0] node[above] {$a$} (9);
	\path (1) edge [loop left,looseness=8, in = 150, out = 210] node[left] {$a$} (1);
	\path (11) edge [loop left,looseness=8, in = -30, out = 30] node[right] {$a$} (11);
	
	\path (1) edge [bend right = 0] node[right] {$b$} (2);
	\path (2) edge [bend right = 0] (1);
	\path (3) edge [bend right = 0] node[above] {$b$} (4);
	\path (4) edge [bend right = 0] node[above] {$b$} (5);
	\path (5) edge [bend right = 0] node[above] {$b$} (4);
	\path (7) edge [bend right = 0] node[right] {$b$} (3);
	\path (8) edge [bend right = 0] node[above] {$b$} (9);
	\path (9) edge [bend right = 0] node[above] {$b$} (8);
	\path (10) edge [bend right = 0] node[above] {$b$} (11);
	\path (11) edge [bend right = 0] node[above] {$b$} (10);
	
	\path (6) edge [loop left,looseness=8, in = -30, out = 30] node[right] {$b$} (6);
	
	\end{tikzpicture}
\end{center}
Automaton on 11 states with switch count 65. The shortest synchronizing word has length 77 and switch count 65:
\[
b^2(ab)^4b(ab)^5b(ab)^6(ba)^3(ab)^4(ba)^3(ab)^3b(ba)^3(ab)^2b^2(ab)^2.
\]

\section{Basic constructions}
\label{secbc}

In this section we give some series of automata for which the switch count is quadratic, more precisely, $\frac{1}{2}n^2 + O(n)$ for automata on $n$ states.
A first one satisfying this property is $F({\cal C}_{n/2})$ for even $n$ as presented in Section \ref{seccomp}, having three symbols and switch count $\frac{n^2 -4n+4}{2}$. Here we want to
present and discuss some other patterns that lead to quadratic switch counts. These constructions provide relatively easy improvements over the bounds in Section \ref{seccomp} and include a sequence of binary automata.

The first construction has switch count $\frac{n^2 - n}{2}$ for automata on $n$ states and $n-1$ symbols. A slight improvement of the same idea yields switch count $\frac{n^2 + n}{2}$ for $n \geq 5$, being the highest switch count that we know for $n = 5,6,7,8,9,10,11$.
The last construction is binary and yields switch count $\frac{n^2 - 6n + 10}{2}$ for $n$ even.

\subsection{No bound on number of symbols}

For $n \geq 2$ we define the DFA ${\cal P}_n$ on $n$ states $1,2,\ldots,n$ and $n-1$ symbols $a_1,a_2,\ldots,a_{n-1}$ as follows:
\[2 a_1 = 1,\;\; q a_1 = q \mbox{ for all $q \neq 2$}\]
and
\[i a_i = i+1,\;\; (i+1) a_i = i,\;\; q a_i = q \mbox{ for all $q \neq i, i+1$}\]
for all $i = 2,3,\ldots,n-1$.
Leaving out all arrows from a state to itself, in a picture this looks as follows:

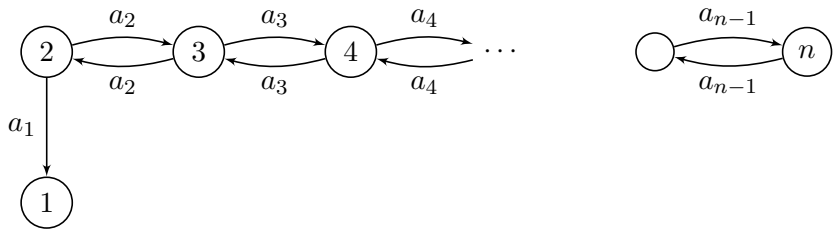
\begin{center}
	\begin{tikzpicture}[-latex',node distance =2 cm and 2cm ,on grid ,
	semithick , state/.style ={ circle ,top color =white , bottom
		color = white!20 , draw, black , text=black , minimum width =.5
		cm}]
	
	\node[state] (2) {2};
	\node[state] (1) [below =of 2] {1};
	\node[state] (3) [right =of 2] {3};
	\node[state] (4) [right =of 3] {4};
	\node (5) [right =of 4] {$\cdots$};
	\node[state] (6) [right =of 5] {};
	\node[state] (7) [right =of 6] {$n$};
	\path (2) edge node[left] {$a_1$} (1);
	
	\path (2) edge [bend left = 15] node[above] {$a_2$} (3);
	\path (3) edge [bend left = 15] node[below] {$a_2$} (2);
	\path (3) edge [bend left = 15] node[above] {$a_3$} (4);
	\path (4) edge [bend left = 15] node[below] {$a_3$} (3);
	\path (4) edge [bend left = 15] node[above] {$a_4$} (5);
	\path (5) edge [bend left = 15] node[below] {$a_4$} (4);
	\path (6) edge [bend left = 15] node[above] {$a_{n-1}$} (7);
	\path (7) edge [bend left = 15] node[below] {$a_{n-1}$} (6);
	
	\end{tikzpicture}
\end{center}

\begin{theorem}
	\label{thm1}
	For every $n\geq 2$ the DFA ${\cal P}_n$ is synchronizing, and its switch count is equal to its minimal synchronizing word length, which is $\frac{n^2 - n}{2}$.
\end{theorem}
\begin{proof}
	As for every symbol $a$ and every $k>1$ the operation $a^k$ is equal to the operation $a$ or the identity, we see that in a minimal synchronizing word never two consecutive equal symbols will occur. The switch count is equal to the minimal synchronizing word length.
	
	For $k = 1,\ldots,n-1$ we define $w_k$ by
	$w_1 = a_1$ and $w_k = w_{k-1} a_k a_{k-1} \cdots a_2 a_1$ for $k = 2,\ldots,n-1$. We prove by induction on $k$ that in the power automaton we have $\{1,2,\ldots,k+1\} w_k = \{1\}$.
	For $k=1$ this holds since $2 a_1 = 1 a_1 = 1$. For the induction step we compute
	\[ \begin{array}{rcl}
	\{1,2,\ldots,k+1\} w_k &=& \{1,2,\ldots,k+1\} w_{k-1} a_k a_{k-1} \cdots a_2 a_1 \\
	&=& \{1,k+1\} a_k a_{k-1} \cdots a_2 a_1 \\
	&=& \{1\} \end{array} \]
	using the induction hypothesis $\{1,2,\ldots,k\} w_{k-1} = \{1\}$ and the fact that all symbols in $w_{k-1}$ act as the identity on $k+1$.
	
	This proves that $w_{n-1}$ is synchronizing, and it has both length and switch count $\frac{n^2 - n}{2}$.
	
	To prove that no synchronizing word of lower switch count exists let $w$ be any synchronizing word. Define the weight $W(q)$ of state $q$ to be $q-1$ for $q = 1,2,\ldots,n$, and the weight $W(S)$ of a set $S$ of states
	to be the sum of the weights of its elements. Then for every set $S$ of states, for every $i,k > 0$  we obtain $W(S a_i^k) \geq W(S) -1$ since $a_i = a_i^k$ moves one state by one position, and
	leaves all other states. As 1 is the only state to which the word $w$ can synchronize, $W(\{1,2,\ldots,n\}) = \frac{n^2 - n}{2}$ and $W(\{1\}) = 0$ we conclude that the length of $w$ is at least $\frac{n^2 - n}{2}$.
	
\end{proof}

The automaton ${\cal P}_n$ can be seen as an instance of a more general idea: take a basic synchronizing automaton, in this case two states 1 and 2 with one symbol $a_1$ mapping both states to 1, and extend a particular state, in this case state 2, by a sequence of fresh states and fresh symbols that behave as symbols 2 to $n$ in ${\cal P}_n$. Other instances of the same idea may lead to slightly higher switch counts, but always still of the shape $\frac{n^2}{2} + O(n)$. For instance,

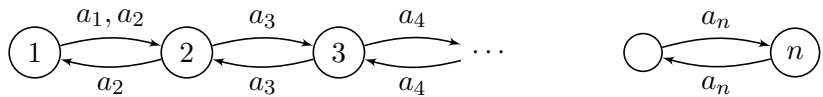
\begin{center}
	\begin{tikzpicture}[-latex',node distance =2 cm and 2cm ,on grid ,
	semithick , state/.style ={ circle ,top color =white , bottom
		color = white!20 , draw, black , text=black , minimum width =.5
		cm}]
	
	\node[state] (2) {1};
	\node[state] (3) [right =of 2] {2};
	\node[state] (4) [right =of 3] {3};
	\node (5) [right =of 4] {$\cdots$};
	\node[state] (6) [right =of 5] {};
	\node[state] (7) [right =of 6] {$n$};
	
	\path (2) edge [bend left = 15] node[above] {$a_1, a_2$} (3);
	\path (3) edge [bend left = 15] node[below] {$a_2$} (2);
	\path (3) edge [bend left = 15] node[above] {$a_3$} (4);
	\path (4) edge [bend left = 15] node[below] {$a_3$} (3);
	\path (4) edge [bend left = 15] node[above] {$a_4$} (5);
	\path (5) edge [bend left = 15] node[below] {$a_4$} (4);
	\path (6) edge [bend left = 15] node[above] {$a_n$} (7);
	\path (7) edge [bend left = 15] node[below] {$a_n$} (6);
	
	\end{tikzpicture}
\end{center}

has switch count $\frac{n^2 + n - 4}{2}$.

The highest value of a variant of this shape we found is obtained by starting from the automaton found by Roman in \cite{R08} on 5 states and 3 symbols having minimal synchronization length 16 (so the same length as ${\cal C}_5$), and switch count 15. This automaton is equal to ${\cal R}_5$ where for $n \geq 5$ the automaton ${\cal R}_n$ is defined by the picture below, again omitting self-loops.

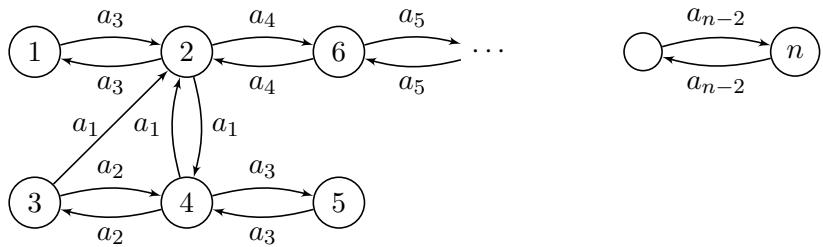
\begin{center}
	\begin{tikzpicture}[-latex',node distance =2 cm and 2cm ,on grid ,
	semithick , state/.style ={ circle ,top color =white , bottom
		color = white!20 , draw, black , text=black , minimum width =.5
		cm}]
	
	\node[state] (1) {1};
	\node[state] (3) [below =of 1] {3};
	\node[state] (2) [right =of 1] {2};
	\node[state] (4) [right =of 3] {4};
	\node[state] (5) [right =of 4] {5};
	\node[state] (6) [right =of 2] {6};
	\node (7) [right =of 6] {$\cdots$};
	\node[state] (8) [right =of 7] {};
	\node[state] (9) [right =of 8] {$n$};
	\path (3) edge node[left] {$a_1$} (2);
	
	\path (1) edge [bend left = 15] node[above] {$a_3$} (2);
	\path (2) edge [bend left = 15] node[below] {$a_3$} (1);
	\path (3) edge [bend left = 15] node[above] {$a_2$} (4);
	\path (4) edge [bend left = 15] node[below] {$a_2$} (3);
	\path (4) edge [bend left = 15] node[above] {$a_3$} (5);
	\path (5) edge [bend left = 15] node[below] {$a_3$} (4);
	\path (2) edge [bend left = 15] node[right] {$a_1$} (4);
	\path (4) edge [bend left = 15] node[left] {$a_1$} (2);
	\path (2) edge [bend left = 15] node[above] {$a_4$} (6);
	\path (6) edge [bend left = 15] node[below] {$a_4$} (2);
	\path (6) edge [bend left = 15] node[above] {$a_5$} (7);
	\path (7) edge [bend left = 15] node[below] {$a_5$} (6);
	\path (8) edge [bend left = 15] node[above] {$a_{n-2}$} (9);
	\path (9) edge [bend left = 15] node[below] {$a_{n-2}$} (8);
	
	\end{tikzpicture}
\end{center}

\begin{theorem}
	For every $n \geq 5$ the automaton ${\cal R}_n$ is synchronizing and  has switch count $\frac{n^2 + n}{2}$.
\end{theorem}
\begin{proof}
	Observe that $\{1,2,3,4,5\} a_1 = \{1,2,4,5\}$ and $\{1,2,4,5\} a_2a_3a_1a_3a_1 = \{1,4,5\}$. So by the string $a_2a_3a_1a_3a_1$ of switch count 5 the state 2 is removed from the set
	$\{1,2,4,5\}$. By doing this $n-4$ times in combination with $a_i$ for $i \geq 4$ similar to the proof of Theorem \ref{thm1}, all states $q \geq 6$ can be removed. Finally, the set $\{1,4,5\}$ synchronizes by $a_3a_2a_3a_1^2a_3a_1a_2a_3a_1$ to the state 2, yielding a total switch count $\frac{n^2 + n}{2}$.
	
	To prove that a lower switch count is not possible, for a set $V \subseteq \{1,2,\ldots,n\}$ we define $S(V)$ to be the smallest switch count of a word $w$ for which $(V \cap \{1,2,3,4,5\})w = \{2\}$, and define the weight $W(V)$ by
	\[ W(V) \; = \; S(V) + \sum_{i \in V \wedge i > 5} i .\]
	Let $Q = \{1,2,\ldots,n\}$. Note that $S(Q) = 15$ is equal to the switch count of the Roman automaton ${\cal R}_5$, in which the synchronizing state is 2. So $W(Q) = 15 + \sum_{i=6}^n i = \frac{n^2 + n}{2}$, and $W(\{2\}) = 0$. It remains to show that if $Va_i = V'$, then $W(V') \geq W(V)-1$, for all $i,V$. For $i\geq 5$ and $i\leq 3$ this is direct from the definition.
	For $i=4$ this follows from the fact that for reachable $V \subseteq \{1,2,3,4,5\}$ it holds that if $2 \in V$ and $Vw \subseteq V \setminus \{2\}$ then $w$ has switch count $\geq 5$. This is proved by analyzing the power automaton of the Roman automaton ${\cal R}_5$.
\end{proof}

\subsection{Restricting to binary automata}

A first idea to find automata on two symbols with a high switch count is to take the automata ${\cal P}_n$ or variants, and identify many of the symbols.
This turns out to yield very low switch counts, so we need something completely different. The modification $F_2({\cal C}_n)$ of ${\cal C}_n$ as given in the end of Section \ref{seccomp} only yields
$\frac{2n^2}{9} + O(n)$ for $n$ being the number of states. Again inspired by ${\cal C}_n$, instead we define the following binary automaton ${\cal Q}_n$ on the even number
$n = 2k$ of states $1,2,3,\ldots,2k$ on which the two symbols $a,b$ are defined by
\[ (2i)a = (2i-1)a = 2i \mbox{ for $i = 1,2,\ldots,k$, and}\]
\[ (2i-2)b = (2i-1)b = 2i-1 \mbox{ for $i = 2,3,\ldots,k$, and } 1b = 3,\; (2k)b = 1,\]
as is shown in the following picture:
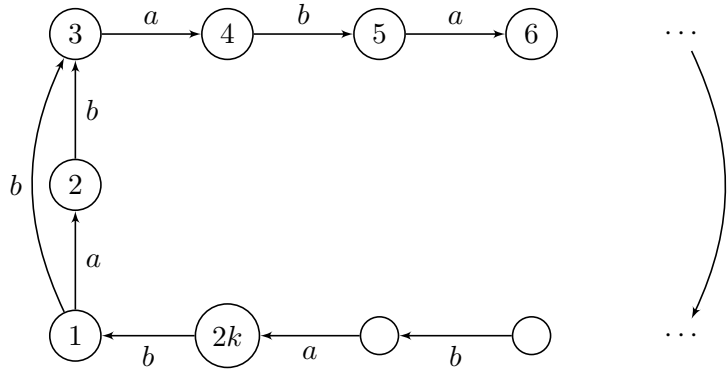
\begin{center}
	\begin{tikzpicture}[-latex',node distance =2 cm and 2cm ,on grid ,
	semithick , state/.style ={ circle ,top color =white , bottom
		color = white!20 , draw, black , text=black , minimum width =.5
		cm}]
	
	\node[state] (3) {3};
	\node[state] (2) [below =of 1] {2};
	\node[state] (1) [below =of 2] {1};
	\node[state] (4) [right =of 3] {4};
	\node[state] (5) [right =of 4] {5};
	\node[state] (6) [right =of 5] {6};
	\node (7) [right =of 6] {$\cdots$};
	\node[state] (8) [right =of 1] {$2k$};
	\node[state] (9) [right =of 8] {};
	\node[state] (10) [right =of 9] {};
	\node (11) [right =of 10] {$\cdots$};
	\path (2) edge node[right] {$b$} (3);
	\path (1) edge node[right] {$a$} (2);
	\path (3) edge node[above] {$a$} (4);
	\path (4) edge node[above] {$b$} (5);
	\path (5) edge node[above] {$a$} (6);
	\path (8) edge node[below] {$b$} (1);
	\path (9) edge node[below] {$a$} (8);
	\path (10) edge node[below] {$b$} (9);
	
	\path (1) edge [bend left = 25] node[left] {$b$} (3);
	\path (7) edge [bend left = 25] (11);
	
	\end{tikzpicture}
\end{center}

in which self-loops are omitted as before.

\begin{theorem}
	\label{thm2}
	For every even $n\geq 4$ the DFA ${\cal Q}_n$ is synchronizing, and its switch count is equal to $\frac{n^2 - 6n + 10}{2}$.
\end{theorem}
\begin{proof}
	First we observe that after one single $b$ step only the odd states occur, and on the odd states $ab$ acts as the cyclic permutation mapping $2i+1$ to $2i+3$ modulo $n$, and $b$ maps 1 to 3, and all other odd states to itself. So $ab$ and $b$ act as the \v{C}ern\'y automaton on the $k$ odd states in ${\cal Q}_n$, where $k = n/2$. Hence $b(b(ab)^{k-1})^{k-2} b$ is a synchronizing word, having switch count $2(k-1)(k-2)+1 = 2(\frac{n}{2}-1)(\frac{n}{2}-2)+1 = \frac{n^2 - 6n + 10}{2}$.
	
	It remains to prove that no synchronizing word of smaller switch count exists. Let $w$ be any synchronizing word of minimal switch count. Since for every $i>0$ the sequence $a^i$ acts the same as one single $a$, we may assume that $w$ does not contain consecutive $a$'s.  The word $w$ starts by a $b$, since
	$ab$ transforms the set of all states to the set of all odd states, just like $b$ does. After this $b$ step in the power automaton only the odd states are involved, and since the rest of $w$ does not contain consecutive $a$'s, it is in the language of the regular expression $(b + ab)^*$, may be followed by an $a$. But both $b$ and $ab$ map every odd state to an odd state, and $a$ is injective on all odd states, so the synchronization does not end in $a$. Hence after the first $b$ only $b$ steps and $ab$ steps occur and the whole synchronization game is only played in the odd states. We already observed that on these odd states the operations $b$ and $c = ab$ act as the \v{C}ern\'y automaton on the $k$ odd states, in which $c$ is the cyclic symbol. So $w = b h(w')$ in which $w'$ is a synchronizing word of this \v{C}ern\'y automaton on $k$ states on the symbols $b,c$, and $h$ is the homomorphism mapping $c$ to $ab$ and $b$ to $b$. Observe that the switch count of $w$ is $2s+1$ where $s$ is the number of $a$'s in $w$, which is equal to the number of $c$'s in $w'$. So it remains to prove the following claim.
	\begin{quote}
		{\bf Claim:} In every synchronizing word $w'$ of the \v{C}ern\'y automaton ${\cal C}_k$ on $k$ states, the number of $c$'s in $w'$ is at least $(k-1)(k-2)$.
	\end{quote}
	Number the states of ${\cal C}_k$ such that $1b = 2$, $ib = i$ for $i \neq 1$ and $ic = i+1$ modulo $k$ for all $i$. Assume that $w'$ ends in $b$ and synchronizes in 2, otherwise remove some elements of $w'$ from the end. Now choose $1 \leq j \leq k-1$ such that the number of $c$'s in $c^j w'$ is divisible by $k-1$. Write $c^j w' = v_1 v_2 \cdots v_p$ such that every $v_i$ is of the shape $(c b^*)^k-1$, so, starts by $c$ and contains exactly $k-1$ $c$'s. For every state $q$ and every word $v$ of this shape it holds that either $qv = q$ (if it contains a $b$ step from 1 to 2) or $qv = q-1$ modulo $k$ (if it does not contain such a step). The key observation here is that in processing such a $v$, at most one $b$ step from 1 to 2 is possible. Further we observe that since $v_1$ starts in $c$, we obtain $1 v_1 = k$. So for the state $1$ to synchronize in 2, it has to do $k-1$ such steps: after the first $v_1$-step it has moved to $k$, and then at least $k-2$ steps corresponding to $v_2,v_3,\ldots,v_p$ are needed to reach state 2. Hence $p \geq k-1$, and the number of $c$'s in $c^j w'$ is exactly $p(k-1) \geq (k-1)(k-1)$. Since $j \leq k-1$ we conclude that the number of $c$'s in $w'$ is at least $(k-1)(k-1) - j \geq  (k-1)(k-2)$, concluding the proof.
\end{proof}

\section{Analysis of the switch count of $\mathcal{A}_n$}
\label{secmr}

In this section we prove that the maximal possible switch count of automata on $n$ states grows at least like $\frac23 n^2$. We prove this by presenting the series of binary automata $\mathcal{A}_n$
and showing that the switch counts of these automata are of this order. We are not aware of any construction (either binary or non-binary) with quadratic switch count and a higher leading constant. For $n=8$ and $n=9$, our experimental results confirm that this construction indeed reaches the maximal possible switch count if we would restrict to binary automata. 

Let $n\geq 3$ and consider the automaton $\mathcal{A}_n = (Q_n,\Sigma_n)$ with state set $Q_n=\{1,\ldots,n\}$ and alphabet $\Sigma_n = \left\{a,b\right\}$ defined by
\begin{align*}
&\begin{array}{lll}
1a = 1, & 1b = 2;&\\
(2k)a = 2k+1, \qquad & (2k)b = 2k-1,& k = 1,\ldots,\lfloor\frac{n-1}{2}\rfloor;\\
(2k+1)a = 2k, & (2k+1)b = 2k+2,\qquad & k = 1,\ldots,\lfloor\frac{n-2}{2}\rfloor;
\end{array}\\
&na = nb = \left\{\begin{array}{ll}
\frac{n}{3}&\text{if $n=0$ (mod $3$)}\\
\frac{n+2}{3}&\text{if $n=1$ (mod $3$)}\\
\frac{n+4}{3}\qquad &\text{if $n=2$ (mod $3$)}
\end{array}\right.
\end{align*}
An illustration of this automaton for $n=7$ was already given in the introduction. Switch counts for $3\leq n\leq 12$ are given in the following table.
\[
\begin{array}{c||c|c|c|c|c|c|c|c|c|c}
n & 3 & 4 & 5 & 6 & 7 & 8 & 9 & 10 & 11 & 12\\
\hline
\ sc(\mathcal{A}_n)\ &\ \ 1\ \  &\ \  5 \ \  &\ \  9 \ \  &\ \  15 \ \  &\ \  23 \ \  &\ \  31 \ \  &\ \  41 \ \  &\ \  53 \ \  &\ \  65 \ \  &\ \  79 \ \
\end{array}
\]

Now we arrive at the main theorem giving the exact switch count of $\mathcal{A}_n$.
\begin{theorem}\label{theoem:lowerbound}
	The DFA $\mathcal{A}_n$ has switch count $\left\lceil  \frac23 n(n-2)-1 \right\rceil$ for all $n\geq 3$.
\end{theorem}
For notational convenience, we will prove this result for $n$ divisible by 6. Other cases are very similar. We work towards the proof of this theorem by defining some tools for the analysis. First we construct an auxiliary automaton $\mathcal{B}_n$ as follows. The state set is $Q_{\mathcal{B}_n}= \bigcup_{q=1}^n \left\{q,-q\right\}$ and the alphabet is $\Sigma_{\mathcal{B}_n} = \Sigma_n = \left\{a,b\right\}$. This means we slightly abuse notation, because we want to be able to apply the same word in both automata. From the context it will be clear which automaton we are talking about.

All transitions from the states $2,\ldots,n-1$ are equal to those in the original automaton $\mathcal{A}_n$. Moreover, we let
\[
1a = -1,\qquad 1b = 2,\qquad na = nb = -\frac{n}{3}.
\]
Finally, we define $(-q)a = -(qa)$ and $(-q)b = -(qb)$ for all $q\in \{1,\ldots,n\}$. See Figure \ref{fig:B} for $\mathcal{B}_n$.

\begin{figure}
	\begin{tikzpicture}[-latex ,node distance =3 cm and 3cm ,on grid ,
	semithick ,
	state/.style ={ circle ,top color =white , bottom color = white!20 ,
		draw, black , inner sep =0pt, text=black , minimum width =.9 cm}]
	\node[state,white] (0) {}; 
	\node[state] (1) at ($(0)+3*({-2*cos(14)},{sin(14)})$) {$\scriptstyle 1$};
	\node[state, top color = gray!40, bottom color = gray!40] (2) at ($(0)+3*({-2*cos(36)},{sin(36)})$) {$\scriptstyle 2$};
	\node[state,white,text=black] (3) at ($(0)+3*({-2*cos(53)},{sin(53)})$)  {$\ldots$};
	\node[state] (4) at ($(0)+3*({-2*cos(70)},{sin(70)})$)  {$\scriptstyle \frac{n}{3}-1$};
	\node[state, top color = gray!40, bottom color = gray!40] (5) at ($(0)+3*({-2*cos(90)},{sin(90)})$)  {$\scriptstyle\frac{n}{3}$};
	\node[state] (6) at ($(0)+3*({-2*cos(110)},{sin(110)})$)  {$\scriptstyle\frac{n}{3}+1$};
	\node[state,white, text=black] (7) at ($(0)+3*({-2*cos(127)},{sin(127)})$)  {$\ldots$};
	\node[state] (8) at ($(0)+3*({-2*cos(144)},{sin(144)})$)  {$\scriptstyle n-1$};
	\node[state, top color = gray!40, bottom color = gray!40] (9) at ($(0)+3*({-2*cos(166)},{sin(166)})$)  {$\scriptstyle n$};
	
	\node[state, top color = gray!40, bottom color = gray!40] (-1) at ($(0)+3*({-2*cos(14)},{-sin(14)})$) {$\scriptstyle -1$};
	\node[state] (-2) at ($(0)+3*({-2*cos(36)},{-sin(36)})$) {$\scriptstyle -2$};
	\node[state,white,text=black] (-3) at ($(0)+3*({-2*cos(53)},{-sin(53)})$)  {$\ldots$};
	\node[state, bottom color = gray!40, top color = gray!40] (-4) at ($(0)+3*({-2*cos(70)},{-sin(70)})$)  {$\scriptstyle -\frac{n}{3}+1$};
	\node[state] (-5) at ($(0)+3*({-2*cos(90)},{-sin(90)})$)  {$\scriptstyle -\frac{n}{3}$};
	\node[state,top color = gray!13, bottom color = gray!13] (-6) at ($(0)+3*({-2*cos(110)},{-sin(110)})$)  {$\scriptstyle -\frac{n}{3}-1$};
	\node[state,white, text=black] (-7) at ($(0)+3*({-2*cos(127)},{-sin(127)})$)  {$\ldots$};
	\node[state,top color = gray!13, bottom color = gray!13] (-8) at ($(0)+3*({-2*cos(144)},{-sin(144)})$)  {$\scriptstyle -n+1$};
	\node[state] (-9) at ($(0)+3*({-2*cos(166)},{-sin(166)})$)  {$\scriptstyle -n$};
	
	\path (1) edge node[above left]{$b$} (2);
	\path (2) edge node[above left]{$a$} (3);
	\path (3) edge node[above]{$a$} (4);
	\path (4) edge node[above]{$b$} (5);
	\path (5) edge node[above]{$a$} (6);
	\path (6) edge node[above]{$b$} (7);
	\path (7) edge node[above right]{$a$} (8);
	\path (8) edge node[above right]{$b$} (9);
	
	\path (8) edge (7);
	\path (7) edge (6);
	\path (6) edge (5);
	\path (5) edge (4);
	\path (4) edge (3);
	\path (3) edge (2);
	\path (2) edge (1);
	\path (1) edge node[left]{$a$}(-1);
	
	\path (-1) edge node[below left]{$b$} (-2);
	\path (-2) edge node[below left]{$a$} (-3);
	\path (-3) edge node[below]{$a$} (-4);
	\path (-4) edge node[below]{$b$} (-5);
	\path (-5) edge node[below]{$a$} (-6);
	\path (-6) edge node[below]{$b$} (-7);
	\path (-7) edge node[below right]{$a$} (-8);
	\path (-8) edge node[below right]{$b$} (-9);
	
	\path (-8) edge (-7);
	\path (-7) edge (-6);
	\path (-6) edge (-5);
	\path (-5) edge (-4);
	\path (-4) edge (-3);
	\path (-3) edge (-2);
	\path (-2) edge (-1);
	\path (-1) edge (1);
	
	\path (9) edge [bend right = 30] node[below]{$a,b$}(-5);
	\path (-9) edge [bend right = -30] node[above]{$a,b$}(5);
	\end{tikzpicture}\caption{The automaton $\mathcal{B}_n$. All shades states are in $S$, darker shaded states in $C$.}\label{fig:B}
\end{figure}
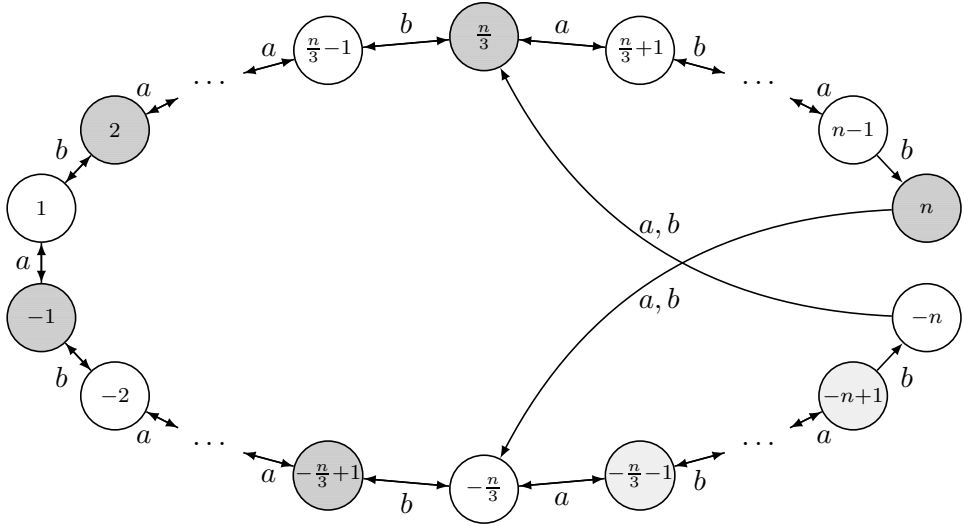

Define the set $S\subseteq Q_{\mathcal{B}_n}$ by
\[
S = \left\{2k:1\leq k\leq \frac{n}{2}\right\}\cup\left\{-2k+1:1\leq k\leq \frac{n}{2}\right\}.
\]
Our goal will be to synchronize the set $S$. Observe that $Sw\subseteq S$ if $|w|$ is even, and $Sw\subseteq S^c$ if $|w|$ is odd. Therefore, from now on we will assume that every subset under consideration is either contained in $S$ or in $S^c$. Furthermore, $w\in\Sigma_n^\star$ is a synchronizing word for $\mathcal{A}_n$ if and only if $|Sw|=1$ in $\mathcal{B}_n$.

For $A\subseteq S$, we define $-A :=\left\{q:-q\in A\right\}$. Then $(-A)w = -(Aw)$ for every word $w$ and every set $A\subseteq S$, so we will write $-Aw$. In particular this means that $|Aw|=|-Aw|$. Because of this symmetry, in some of the results that follow it suffices to study subsets of $S$, for subsets of $S^c=-S$ analogous conclusions then follow instantly and will not be explicitly stated.

For $x,y\in S$, define $[x,y]=\left\{q\in S:x\leq q\leq y\right\}$. For $x,y\in -S$, we let $[x,y]=\left\{q\in -S:x\leq q\leq y\right\}$. These sets will be called \emph{intervals}.
An important set is the interval $C:=[-\frac{n}{3}+1,n]$, on which we will make some remarks to quickly derive a synchronizing word for $\mathcal{A}_n$ and get a better understanding. First, the word $ab$ defines a cyclic permutation on $C$. A closer look reveals that $b^2$ is the identity on $C$, except on the state $n$, where we have $nb^2 = -\frac{n}{3}+1$. So we observe that actually the \v{C}ern\'y automaton is hidden here. Since $|C|=\frac23 n$, a word that synchronizes $C$ is $b^2((ab)^{\frac23 n-1}b^2)^{\frac23n-2}$. Using the fact that $S(ab)^{\frac{n}{3}}=C$, we find the following synchronizing word for $S$ and hence for the automaton $\mathcal{A}_n$:
\[
w = (ab)^{\frac{n}{3}}b^2((ab)^{\frac23 n-1}b^2)^{\frac23n-2}.
\]
This word has switch count $|w|_s = \frac89n^2+\frac23n$. As we will see, synchronizing words with smaller switch count exist, but some patterns are similar and the cycle $C$ will play a key role in the analysis.

The following lemma gives one reason why $C$ is important. As can be seen in Figure \ref{fig:B}, $a$ and $b$ can only rotate sets that do not contain $n$ or $-n$. To make progress, we need to visit sets that do contain $n$ or $-n$. The lemma essentially  tells that such sets will be contained in $C$ (or $-C$).

\begin{lemma}\label{lem:subsetC}
	Suppose $A\subseteq C$ and $w\in\Sigma_n^\star$. If $n\in Aw$, then $Aw\subseteq C$. If $-n\in Aw$, then $Aw\subseteq -C$.
\end{lemma}

\begin{proof}
	Define intervals $I_x = [x,x+\frac43 n-1]$ for $-n\leq x\leq -\frac{n}{3}+1$. This means that $I_x\subseteq S$ if $x$ odd and $I_x\subseteq -S$ if $x$ even. We will first prove the following claim. If $v\in\Sigma_n^\star$ and $B \subseteq I_x$ for some $x$, then $Bv \subseteq I_y$ for some $y$. It suffices to prove this for $|v|=1$, since we can lift the property to longer words by induction. First assume $-n+1\leq x\leq -\frac{n}{3}$. In this case $B$ does not contain $n$ or $-n$ so that $a$ and $b$ will rotate $B$. Thus, if $x$ is even, then $Ba\subseteq I_{x-1}$ and $Bb\subseteq I_{x+1}$; if $x$ is odd, then $Ba\subseteq I_{x+1}$ and $Bb\subseteq I_{x-1}$. To conclude, we consider the two border cases. If $x=-n$, then $Ba\subseteq I_{x+1}$ and $Bb\subseteq I_{x+1}$. If $x=-\frac{n}{3}+1$, then $Ba\subseteq I_{x-1}$ and $Bb\subseteq I_{x-1}$, proving the claim.
	
	Now let $A\subseteq C=I_{-\frac{n}{3}+1}$ so that by the claim $Aw\subseteq I_x$ for some $-n\leq x\leq -\frac{n}{3}+1$. If $Aw\subseteq I_x$ for $x\neq -\frac{n}{3}+1$, then $n\not\in Aw$. Otherwise $Aw\subseteq I_{-\frac{n}{3}+1}= C$. The second statement follows analogously.
\end{proof}

The core ingredient of our approach will be a measure on subsets of $S$. Before introducing this measure, we define and study distances between states. The measure of a set will be based on this notion of distance.

\begin{definition}
	For $p,q\in S$, define the (asymmetric) distance from $p$ to $q$ by
	\[
	d(p,q) = \min\left\{k>1:p(ab)^{\frac{n}{3}+k}=q(ab)^\frac{n}{3}\right\}.
	\]
\end{definition}
For $p,q >-\frac{n}{3}$, this can be interpreted as a clockwise distance on the cycle $C$. For every state $q\in S$ it holds that $q(ab)^\frac{n}{3}$ is on this cycle, so we can measure the distance for $p$ or $q$ smaller than $-\frac{n}{3}$ as well. For $p,q\not\in S$, we define $d(p,q) = d(-q,-p)$, which fits the interpretation of clockwise distance. Some properties of the distance function are listed in the following lemma, where we write $p\simeq q$ if $p(ab)^{\frac{n}{3}}=q(ab)^{\frac{n}{3}}$.

\begin{lemma}\label{lem:distprop} Let $p,q,r\in S$ be such that $p\not\simeq q$ and $d(p,q)<d(p,r)$. Then
	\begin{enumerate}
		\item\label{distprop1} $0< d(p,q)<\frac23 n$;
		\item\label{distprop2} $d(p,q)+d(q,p)=\frac23 n$;
		\item\label{distprop3} $d(q,q)=\frac23 n$;
		\item\label{distprop4} $d(p,q)+d(q,r)=d(p,r)$.
	\end{enumerate}
\end{lemma}
\begin{proof}
	These properties follow straight from the definitions.
\end{proof}

One thing worth noting is that $p\not\simeq q$ is equivalent to $p\neq q$ for $p,q\in C$. We will study how distances change when reading words for ordered pairs $(p,q)\subseteq S$ with $p\neq q$. For all such pairs,
\begin{align}\label{eq:apair}
d(pa,qa) &= d(p,q).
\end{align}
Furthermore, for such pairs we also have
\begin{align}
d(pb,qb) &= \left\{\begin{array}{lll}
d(p,q)&\text{if}&p\neq n,q\neq n,\\
d(p,q)-1\qquad&\text{if}&p=n \text{ and } d(p,q)\geq 2,\\
d(p,q)+1&\text{if}&q=n \text{ and } d(p,q)\leq\frac23 n-1.
\end{array}\right.\label{eq:bpair}
\end{align}
There are two exceptional cases, namely
\begin{align}
d(pb,qb) &= \left\{\begin{array}{lll}
\frac23 n\qquad &\text{if}&p= n \text{ and } d(p,q) = 1,\\
1&\text{if}&q=n\text{ and }d(p,q) = \frac23 n.
\end{array}\right.\label{eq:bpairexcept}
\end{align}
The first exception concerns the pair $(p,q)=(n, -\frac{n}{3}+1)$. In this case the distance jumps from $1$ to $\frac23 n$, since $pb=qb = -\frac{n}{3}$ so that $d(pb,qb) = \frac23n$ (by Lemma \ref{lem:distprop}(\ref{distprop3})). Note that the reversed pair $(p,q)=( -\frac{n}{3}+1,n)$ has distance $\frac23 n-1$ (by Lemma \ref{lem:distprop}(\ref{distprop2})), so that this pair satisfies the second line of (\ref{eq:bpair}). The second exception is the pair $(p,q) = (-\frac{n}{3}-1,n)$, where the distance jumps from $\frac23 n$ to $1$.

The distance can increase and decrease by 1 and can jump in the boundary cases. The subsequent small result shows something extra for pairs in $C$. If we start with a pair in $C$ and increase the distance to $\frac23 n$, then actually the pair is synchronized, so that the distance will never decrease again.

\begin{lemma}\label{lem:Cpairsync}
	Let $p,q\in C$ and let $w$ be a word. If $d(pw,qw)=\frac23 n$, then $pw=qw$.
\end{lemma}

\begin{proof} This is trivial if $p=q$, so assume $p\neq q$. Let $v$ be a strict (possibly empty) prefix of $w$ and let $x$ be the symbol directly after $v$ so that $vx$ is a prefix as well. Choose $v$ and $x$ in such a way that $d(pv,qv)\neq\frac23 n$ and $d(pvx,qvx) = \frac23 n$. Since $x$ changes the distance, we conclude from (\ref{eq:apair}), (\ref{eq:bpair}) and (\ref{eq:bpairexcept}) that $x=b$ and that $n\in \left\{pv,qv\right\}$ (or $-n\in \left\{pv,qv\right\}$), so we know by Lemma \ref{lem:subsetC} that $pv,qv\in C$ (or $-C$). Then also $pvb,qvb\in C$. The fact that $d(pvb,qvb) = \frac23 n$ implies that $pvb\simeq qvb$ (by Lemma \ref{lem:distprop}), which in turn implies that $pvb=qvb$. Since $vb$ is a prefix of $w$, we conclude that $pw=qw$.
\end{proof}

The distance between states $p$ and $q$ can only change if $n\in\left\{p,q\right\}$ (or $-n\in\left\{p,q\right\}$). If $-n+2\leq p,q \leq n-2$, then $\left\{p,q\right\}a^2 = \left\{p,q\right\}b^2=\left\{p,q\right\}$. So to change the distance for such pairs, we have to alternate (odd powers of) $a$ and $b$ until we find a pair of different shape, actually until $n$ or $-n$ is in the image. This alternation can be seen as a rotation in Figure \ref{fig:B}. If we want to change the distance with a word of minimal length, then also the direction of rotation is determined. For instance if $\left\{p,q\right\} = \left\{2,\frac{n}{3}\right\}$, we have to rotate clockwise, since otherwise we would reach the set $-\left\{p,q\right\}$ and by symmetry of the automaton, the word length would not be minimal. This observation is formalized in the following property.

\begin{property}\label{prop:rotprop}
	Let $A = \left\{p,q\right\}$ such that $|q|>|p|$. Suppose that $w$ is such that $d(pw,qw)\neq d(p,q)$ and has minimal switch count. Additionally, assume $w$ is the shortest such word.
	\begin{enumerate}
		\item If $n-|q|$ is even, then $(ab)^{(n-|q|)/2}$ is a prefix of $w$.
		\item If $n-|q|$ is odd, then $(ba)^{(n-|q|-1)/2}b$ is a prefix of $w$.
	\end{enumerate}
\end{property}

By reading one letter, in almost all cases the distance changes by at most 1. In the next lemma, we show that in general longer words are needed to increase the distance by 2. We give lower bounds for the switch count of such words.

\begin{lemma}\label{lem:pairdist}  Let $w\in\Sigma_n^\star$ be a word and $p,q\in C$.
	\begin{enumerate}
		\item Choose $2\leq k\leq \frac23 n-2$. If $d(p,q)\leq k-1$ and $d(pw,qw)=k+1$, then
		\[
		|w|_s \geq \left\{
		\begin{array}{ll}
		\frac23 n+2k-1\qquad & 2\leq k\leq \frac{n}{3},\\
		2n-2k+1\qquad & \frac{n}{3}+1\leq k\leq \frac23 n-2.
		\end{array}
		\right.
		\]
		\item Choose $k = \frac23 n-1$. If $d(p,q)=k-1$ and $d(pw,qw)=k+1$, then
		\[
		|w|_s \geq 2n-2k+1.
		\]
	\end{enumerate}
\end{lemma}

\begin{proof}
	Fix $2\leq k\leq \frac23 n-1$ and choose $p,q$ and $w$ as in the lemma in such a way that $|w|_s$ attains its minimal possible value. Under these conditions, assume that $|w|$ is minimal as well. For $1\leq l\leq |w|$, let $w_l$ be the prefix of length $l$.
	
	First assume $2\leq k\leq \frac23 n-2$, so that $d(pw,qw)\neq \frac23 n$. Then for all prefices $d(pw_l,qw_l)\neq \frac23 n$, by Lemma \ref{lem:Cpairsync}. Consequently, the jumps mentioned in (\ref{eq:bpairexcept}) do not occur so that $|d(pw_l,qw_l)-d(pw_{l+1},qw_{l+1})|\leq 1$ for all $1\leq l < |w|$. Furthermore, $|d(pw_l,qw_l)-d(pw_{l+1},qw_{l+1})|=0$ if $n\not\in\left\{pw_l,qw_l\right\}$, so that by Lemma \ref{lem:subsetC} the distance can only change if $\left\{pw_l,qw_l\right\}\subseteq C$. That means the minimal switch count of $w$ is attained if $d(p,q)=k-1$ and already the first letter of $w$ increases the distance, i.e. $d(pw_1,qw_1) = k$. It follows from (\ref{eq:bpair}) that $w_1=b$ and $q=n$, and since $d(p,q)$ is known to be $k-1$, we can also identify $p$:
	\[
	p = \left\{
	\begin{array}{ll}
	n-2k+2&k = 2,\ldots,\frac{n}{2},\\
	n-2k+1\qquad&k = \frac{n}{2}+1,\ldots,\frac23 n-2.
	\end{array}
	\right.
	\]
	We will further build the word $w$, by keeping track of the images $d(pw_l,qw_l)$ and using properties of pairs. We will subdivide the range $2\leq k\leq \frac23 n-2$ into three cases.
	
	Consider the case $2\leq k\leq \frac{n}{3}-1$. Then $pb = n-2k+1 >\frac{n}{3}$ and $qb = -\frac{n}{3}$. By Property \ref{prop:rotprop}, we have to rotate $\left\{pb,qb\right\}$ clockwise until $n$ is in the image, so we find that $b(ba)^{k-1}b$ is a prefix of $w$. Then we obtain $pb(ba)^{k-1}b = n$ and
	\[
	qb(ba)^{k-1}b= \left\{
	\begin{array}{ll}
	-\frac{n}{3}+2k-1\qquad & k = 2,\ldots,\frac{n}{6},\\
	-\frac{n}{3}+2k\qquad & k = \frac{n}{6}+1,\ldots,\frac{n}{3}-1.
	\end{array}
	\right.
	\]
	If the next letter would be $b$, by (\ref{eq:bpair}) the distance would decrease again to $k-1$, contradicting the minimality of $|w|_s$. Therefore, $b(ba)^k$ is a prefix of $w$, keeping the distance equal to $k$ and giving $pb(ba)^k = -\frac{n}{3}$ and
	\[
	qb(ba)^k= \left\{
	\begin{array}{ll}
	-\frac{n}{3}+2k & k = 2,\ldots,\frac{n}{6}-1,\\
	-\frac{n}{3}+2k+1\qquad & k = \frac{n}{6},\ldots,\frac{n}{3}-1.
	\end{array}
	\right.
	\]
	Then $qb(ba)^k <\frac{n}{3}$, so by Property \ref{prop:rotprop}, this time we have to rotate counterclockwise until $-n$ is in the image. Hence the prefix of $w$ extends to $b(ba)^k(ab)^{\frac{n}{3}}$. This leads to $pb(ba)^k(ab)^{\frac{n}{3}} = -n$, $qb(ba)^k(ab)^{\frac{n}{3}} = -n+2k$ and distance $d(-n,-n+2k) = d(n-2k,n) = k$. Extending the prefix with one more $b$ gives distance $k+1$, so that $w = b(ba)^k(ab)^{\frac{n}{3}}b$ and $|w|_s = \frac{2}{3}n+2k-1$.
	
	The action of the word $w$ could be symbolically summarized as $w = b\circlearrowright a \circlearrowleft b$. For other values of $k$, the analysis is very similar, so we only give the symbolic summary.
	\begin{itemize}
		\item If $k=\frac{n}{3}$, then $w = b\circlearrowright a\circlearrowright b = b(ba)^{\frac23 n-1}b^2$ giving $|w|_s = \frac43 n-1.$
		\item If $\frac{n}{3}+1\leq k\leq \frac23 n-2$, then $w = b\circlearrowleft a\circlearrowleft b = (ba)^{n-k}b^2$ giving $|w|_s = 2n-2k+1$.
	\end{itemize}
	This completes the proof of the first statement of the lemma.
	
	For the second statement, take $k=\frac23 n -1$. In this case $d(p,q) = \frac23 n-2$ and $d(pw,qw)=\frac23 n$. By minimality of $w$, there are no strict prefices for which $d(pw_l,qw_l)=\frac23 n$ and therefore no distance jump from $\frac23 n$ to $1$. However, in view of (\ref{eq:bpairexcept}), we can not immediately conclude that $|d(pw_l,qw_l)-d(pw_{l+1},qw_{l+1})|\leq 1$ for all $1\leq l < |w|$, as there could be a jump from $1$ to $\frac23 n$. Therefore, we split into two cases.
	
	Suppose there is indeed a jump from $1$ to $\frac23 n$. Before this jump occurs, the distance has to change in steps of size $1$ and eventually reach $1$. In particular, this implies existence of a prefix $w_l$ of $w$ for which $d(pw_l,qw_l) = \frac23 n-4$. However, by Lemma \ref{lem:distprop}(\ref{distprop2}), this implies that $d(q,p) = 2$ and $d(qw_l,pw_l)=4$. The word length of $w_l$ can be lower bounded by taking $k=3$ in the first statement of the lemma. This gives $|w|_s\geq |w_l|_s \geq \frac23 n+5$.
	
	If such a jump does not exist, then we are back in the situation where $|d(pw_l,qw_l)-d(pw_{l+1},qw_{l+1})|\leq 1$ for all $1\leq l < |w|$. In the notation as before, we obtain $w = b\circlearrowleft a\circlearrowleft b = (ba)^{\frac{n}{3}+1}b^2$ giving $|w|_s = \frac23 n+3$, which is the minimal possible value.
\end{proof}

Loosely speaking, the previous lemma tells how many pairs of states with mutual distance $k$ have to be visited if we want to increase the distance from $1$ to $\frac23 n$ in steps of size $1$. As a function of $k$, this number first increases and then decreases. Our next goal is to prove that the switch count of $\mathcal{A}_n$ is obtained by summation of this function over $k$ and adding some lower order terms. In order to reach this goal, we have to understand not only how words act on pairs, but also how they act on more general sets.


If a word $w$ synchronizes a set $A\subseteq S$, then it synchronizes each pair in $A$, i.e. for each pair it increases the distance to $\frac23 n$. This means that the switch count of $w$  is determined by distances between pairs. To further investigate this relation between pairs and general sets $A\subseteq S$, we define the measure of a set $A$ as follows.

\begin{definition}
	For $A\subseteq S$, define the measure $\mu(A)$ by
	\[
	\mu(A) = \max_{p\in A}\min_{q\in A} d(p,q).
	\]
\end{definition}
The interpretation is that this measures the largest gap in $A(ab)^{\frac{n}{3}}$ on the cycle $C$. For $A\subseteq S^c$, we use the same definition, so that $\mu(-A)=\mu(A)$. Since we will only consider sets that are either contained in $S$ or $S^c$, we will not define $\mu(A)$ for other subsets of $Q_{\mathcal{B}_n}$.

The maximal possible measure of a set $A$ is $\frac23 n$, which is attained if $A$ is a singleton or a pair $\left\{p,q\right\}$ for which $p\simeq q$. If $A\subseteq C$ has measure $\frac23 n$, then $A$ is a singleton. Since all singletons have measure $\frac23 n$, a word for which $|Sw|=1$ increases the measure from 1 to $\frac23 n$. The minimal possible measure of a set is 1.  Sets with measure 1  are $S$ and $-S$, but some other sets have measure 1 as well. In particular, if $C\subseteq A$, then $\mu(A)=1$.

For pairs $A=\left\{p,q\right\}$, we find
\[
\mu(A) = \max\left\{\min\left\{d(p,p),d(p,q)\right\},\min\left\{d(q,p),d(q,q)\right\}\right\} = \max\left\{d(p,q),d(q,p)\right\}.
\]
The main idea of our argument will be that increasing the measure requires long words and can only be done by steps of size 1. In the next lemma we study the word length needed to increase the measure from 1 to 2.

\begin{lemma}\label{lem:prefix}
	Let $w$ be such that $|Sw|=1$ and $w$ has minimal switch count. Assume $w$ is the shortest such word. Let $v$ be the longest prefix of $w$ such that $\mu(Sv)=1$. Then
	$b(ab)^{\frac{n}{3}-1}$ is a prefix of $v$.
\end{lemma}

\begin{proof} First note that $\mu(S)=1$ and $\mu(Sw)=\frac23 n$. Observe that $Sa = [-n+2,n-1]$ and $Sb = [-n,n-3]$. Further note that $Sa^2=-Sa$ and $Sab=-Sb$. This means that by minimality $w$ does not start with $a^2$ or $ab$, so it starts with $b$.
	
	In the spirit of Property \ref{prop:rotprop}, to synchronize intervals $[x,y]$ not containing $n$ or $-n$, we have to rotate by alternating $a$ and $b$ until $n$ or $-n$ is in the image. The direction of rotation is clockwise if $|y|>|x|$ and counterclockwise otherwise. Now suppose we look for a word $v_k$ with minimal switch count to synchronize the set $[-n,n-k]$ for some odd $k$ satisfying $3\leq k\leq \frac23 n-1$. Noting that
	\[
	[-n,n-k]a^2b=[-n,n-k]b,\qquad\text{and}\qquad [-n,n-k]a^3=[-n,n-k]a,
	\]
	we can assume that does not start with $a^2$. Similarly, it does not start with $b^2$ since
	\[
	[-n,n-k]b^2a=[-n,n-k]a,\qquad\text{and}\qquad [-n,n-k]b^3=[-n,n-k]b.
	\]
	So $v_k$ starts either with $ab$ or with $ba$. Now
	\[
	[-n,n-k]ab = [-n,n-k-2],\qquad\text{and}\qquad [-n,n-k]ba = [-n+4,n-k+2],
	\]
	and the latter turns out to be of the shape $[x,y]$ as above. So if $v_k$ would start with $ba$, we would have to alternate $a$ and $b$ until $n$ or $-n$ is in the image. This means $v_k$ starts with $ba(ab)^2$ or with $ba(ba)^{(-k+3)/2}b$, giving image $[-n,n-k-2]$ and $-[-n,n-k-2]$ respectively. So if $v_k$ would start with $ba$, the switch count would fail to be minimal. Hence $v_k$ starts with $ab$.
	
	Iterating this argument, we conclude that $w$ starts with $u = b(ab)^{\frac{n}{3}-1}$. Finally, observe that for every prefix $\tilde u$ of $u$, we have $\mu(S\tilde u)=1$, so that $u$ is a prefix of $v$.
\end{proof}

\begin{corollary}\label{cor:prefix}
	Let $w,v$ be as in Lemma \ref{lem:prefix}. Then $|v|_s \geq \frac23 n-1$.
\end{corollary}

The following lemma is the key to control which word length is needed to increase the measure of a general set $A$. Namely, an increase in measure can only achieved by increasing the distance between some $p,q\in A$ by at least the same amount.

\begin{lemma}\label{lem:setpair}
	Let $A\subseteq S$ and let $w$ be a word. Then there exist $p,q\in A$ such that
	\[
	d(p,q)\leq \mu(A)\qquad\text{and}\qquad d(pw,qw) = \mu(Aw).
	\]
	If in addition $|Aw|=1$, then there exist $p,q\in A$ such that
	\[
	d(p,q)= \mu(A)\qquad\text{and}\qquad d(pw,qw) = \mu(Aw) = \frac23 n.
	\]
	
\end{lemma}

\begin{proof}
	We will prove the following claim: if $A\subseteq S$ and $p_1,q_1\in Aw$ are such that
	\begin{equation}\label{eq:claim}
	d(p_1,q_1) = \min_{q\in Aw}d(p_1,q),
	\end{equation}
	then there exist $p_0,q_0\in A$ such that $p_0w=p_1$, $q_0w=q_1$ and
	\begin{equation}\label{eq:claimconcl}
	d(p_0,q_0) = \min_{q\in A}d(p_0,q).
	\end{equation}
	This claim implies the first statement of the lemma, since we can take $p_1$ such that $\min_{q\in Aw}d(p_1,q)$ is maximal, i.e. $d(p_1,q_1)=\mu(Aw)$. Clearly we also have $d(p_0,q_0)\leq \mu(A)$. We will take $p_1,q_1$ as above so that in particular (\ref{eq:claim}) is satisfied, and prove this claim by induction on the length of $w$.
	
	First we will prove the claim for $|w|=1$, which is the most involved part of the proof. We distinguish several cases, but in most cases the argument is similar. Namely, for given $p_1$ and $q_1$ we indicate how to choose $p_0$ and $q_0$. We then assume that (\ref{eq:claimconcl}) is false, and show that this leads to a contradiction with (\ref{eq:claim}).
	
	If $w=a$ or $n\not\in A$, then $d(pw,qw) = d(p,q)$ for all $p,q\in A$ by (\ref{eq:apair}) and (\ref{eq:bpair}). Choose $p_0,q_0\in A$ such that $p_0w=p_1$ and $q_0w = q_1$. If there would be $q\in A$ such that $d(p_0,q)<d(p_0,q_0)$, then we would have $d(p_1,qw)<d(p_1,q_1)$, contradicting (\ref{eq:claim}). Therefore, $d(p_0,q_0) = \min_{q\in A}d(p_0,q)$.
	
	From now on, assume $w = b$ and $n\in A$. If $A = \left\{n\right\}$, then the claim clearly holds. If $A = \left\{-\frac{n}{3}+1,n\right\}$, then $Ab = \left\{-\frac{n}{3}\right\}$ so that $p_1=q_1 = -\frac{n}{3}$. Taking $p_0 = -\frac{n}{3}+1$ and $q_0=n$ gives $p_0w = p_1$, $q_0w = q_1$ and $d(p_0,q_0) = \min_{q\in A}d(p_0,q)$. So we can in addition assume that $A\setminus\left\{-\frac{n}{3}+1,n\right\} \neq\emptyset$. This means in particular that $|Aw|\geq 2$ and that $\min_{q\in Aw}d(p_1,q)<\frac23 n$ for every $p_1\in Aw$.
	
	Since $n\in A$, we have $nb = -\frac{n}{3}\in Aw$. It is not possible that $p_1=q_1=nb$, since then $d(p_1,q_1)=\frac23 n$ and (\ref{eq:claim}) would not hold. We split the remaining possibilities into three cases:
	\begin{itemize}
		\item If $q_1 = nb$, then $p_1\neq nb$ and there is a unique $p_0\in A$, $p_0\neq n$ such that $p_0b=p_1$. Choose $q_0=n$. If there would be $q\in A$ such that $d(p_0,q)<d(p_0,n)$, we would by (\ref{eq:bpair}) have $d(p_1,qb) = d(p_0,q)<d(p_0,n)=d(p_1,q_1)-1$, contradicting (\ref{eq:claim}).
		\item If $p_1=nb$, then $q_1\neq nb$ and there is a unique $q_0\in A$, $q_0\not\in nb^{-1}=\left\{-\frac{n}{3}+1,n\right\}$ such that $q_0b=q_1$. There are two subcases:
		\begin{itemize}
			\item	If $-\frac{n}{3}+1\in A$, choose $p_0 = -\frac{n}{3}+1$ so that $d(n,p_0) = 1$ and $p_0b=p_1$. Since $p_0\not\simeq q_0$, we have $d(p_0,q_0)\leq \frac23 n-1$. 	 	
			If there would be  $q\in A$ such that $d(p_0,q)<d(p_0,q_0)$, we would have $d(p_0,q)\leq \frac 23 n-2$. As Lemma \ref{lem:distprop}(\ref{distprop2}) gives $d(p_0,n) = \frac23 n-d(n,p_0) = \frac23 n-1$, it follows that $q\neq n$.  Using (\ref{eq:bpair}), this means $d(p_1,qb) = d(p_0b,qb) = d(p_0,q)<d(p_0,q_0)=d(p_0b,q_0b)=d(p_1,q_1)$, contradicting (\ref{eq:claim}).
			\item If $-\frac{n}{3}+1\not\in A$, choose $p_0=n$. If there would be $q\in A$ such that $d(p_0,q)<d(p_0,q_0)$, by (\ref{eq:bpair}) we would have $d(p_1,qb)=d(p_0b,qb) = d(p_0,q)-1<d(p_0,q_0)-1=d(p_1,q_1)$, contradicting (\ref{eq:claim}).
		\end{itemize}
		
		\item If $p_1\neq nb\neq q_1$, there exist unique $p_0,q_0\in A$ such that $p_0b=p_1$ and $q_0b=q_1$. Also $p_0,q_0\not\in\left\{-\frac{n}{3}+1,n\right\}$. Suppose $q\in A$ is such that $d(p_0,q)<d(p_0,q_0)$. We consider two subcases:
		\begin{itemize}
			\item If $q=n$, then $d(q,q_0)\geq 2$, since $q_0\neq -\frac{n}{3}+1$. Therefore, $d(p_0,q) = d(p_0,q_0)-d(q,q_0)\leq d(p_0,q_0)-2$, using Lemma \ref{lem:distprop}(\ref{distprop4}). It follows that $d(p_1,qb) = d(p_0b,qb) = d(p_0,q)+1 \leq d(p_0,q_0)-1 = d(p_1,q_1)-1<d(p_1,q_1)$, contradicting (\ref{eq:claim}).
			\item If $q\neq n$, then $d(p_1,qb)=d(p_0b,qb)=d(p_0,q)<d(p_0,q_0)=d(p_1,q_1)$, again contradicting (\ref{eq:claim}).
		\end{itemize}
	\end{itemize}	
	In all subcases the assumption that there exists $q\in A$ for which $d(p_0,q)<d(p_0,q_0)$ leads to a contradiction. Therefore, we conclude that $d(p_0,q_0) = \min_{q\in A}d(p_0,q)$, completing the proof for the case $|w|=1$.
	
	Now suppose the claim is true for $|w|=k$. Let $A\subseteq S$ and let $w$ be a word with length $|w|=k+1$. Let the first letter be $u$, so that $w=uv$ for $|v|=k$. Further, let $p_1,q_1\in Aw$ be such that $d(p_1,q_1) = \min_{q\in Aw}d(p_1,q)$. By the induction hypothesis, there exist $\tilde p_1,\tilde q_1\in Au$ such that $\tilde p_1v=p_1$, $\tilde q_1v=q_1$ and $d(\tilde p_1,\tilde q_1)=\min_{q\in Au} d(\tilde p_1,q)$. Since the claim is true for words of length 1, there exist $p_0,q_0\in A$ such that $p_0u=\tilde p_1$, $q_0u=\tilde q_1$ and $d(p_0,q_0)=\min_{q\in A}d(p_0,q)$. Clearly, we also have $p_0w=p_1$ and $q_0w=q_1$, proving the claim for $|w|=k+1$ and hence for all lengths. This establishes the first statement of the lemma.
	
	We proceed to prove the second statement. Obviously, there have to exist $\tilde p$ and $\tilde q$ for which $d(\tilde p,\tilde q)=\mu(A)$. Observe that in this case $pw=qw$ and $d(pw,qw)=\mu(Aw)=\frac23 n$ for all $p,q\in A$, so that in particular this is true for $p=\tilde p$ and $q=\tilde q$, which completes the proof.
\end{proof}

For distances between states, we have seen in (\ref{eq:bpairexcept}) that jumps from $1$ to $\frac23n$ might occur. The following lemma excludes this behavior for measures, so that we can conclude that measures have to increase in steps of size 1.

\begin{lemma}\label{lem:measurediff}
	For all $A\subseteq S$,
	\[
	\mu(Aa) = \mu(A)\qquad\text{and}\qquad  \mu(Ab)\leq \mu(A)+1.
	\]
	Moreover, if $n\not\in A$, then
	\[
	\mu(Ab)=\mu(A).
	\]
\end{lemma}

\begin{proof}
	Let $A\subseteq S$. From (\ref{eq:apair}) and (\ref{eq:bpair}), we directly deduce that $\mu(Aa)=\mu(A)$ and that $\mu(Ab)=\mu(A)$ if $n\not\in A$. If $\left\{n,-\frac{n}{3}+1\right\}\not\subseteq A$, then similarly from (\ref{eq:bpair}) and (\ref{eq:bpairexcept}) it follows that $\mu(Ab)\leq\mu(A)+1$. So assume $\left\{n,-\frac{n}{3}+1\right\}\subseteq A$.
	
	If $\left\{n,-\frac{n}{3}+1\right\} = A$, then
	\[
	\mu(A) = \max\left\{d\left(n,-\frac{n}{3}+1\right),d\left(-\frac{n}{3}+1,n\right)\right\} = \max\left\{1,\frac23n-1\right\} = \frac23n-1.
	\]
	Furthermore, $Ab$ is the singleton $-\frac{n}{3}$ so that $\mu(Ab)=\frac23n$, satisfying the lemma.
	
	Finally, suppose that $A\setminus\left\{-\frac{n}{3}+1,n\right\}$ is nonempty. Let $p\in A$. If $p\neq n$, then $d(pb,qb)\leq d(p,q)+1$ for all $q\in A$. Hence
	\[
	\min_{q\in A}d(pb,qb) \leq \min_{q\in A}d(p,q)+1.
	\]
	If $p=n$, let $r\in A$ be such that $d(pb,rb) = d(-\frac{n}{3},rb)$ is minimal. For $q\in A\setminus\left\{-\frac{n}{3}+1,n\right\}$ we have $qb\not\simeq -\frac{n}{3}$ and consequently $d(pb,qb)\leq\frac23 n-1$. In particular this implies $d(pb,rb)\leq\frac23 n-1$. Then, using (\ref{eq:bpair}),
	\[
	\min_{q\in A}d(pb,qb) = d(nb,rb) = d(n,r) - 1 = d\left(-\frac{n}{3}+1,r\right) = \min_{q\in A}d\left(-\frac{n}{3}+1,q\right).
	\]
	So for all $\tilde p\in A$, we find
	\[
	\min_{q\in A}d(\tilde pb,qb) \leq \max_{p\in A}\min_{q\in A}d(p,q)+1.
	\]
	Taking the maximum over $\tilde p$ gives $\mu(Ab)\leq \mu(A)+1$.
\end{proof}

The measure of a set can only increase if there is a distance change for some pair $\left\{p,q\right\}$ in the set. The following result states that every time we increase the measure of a set, this set was contained in $C$ or $-C$. This guarantees that Lemma \ref{lem:pairdist} applies to the pairs for which the distance changes.

\begin{lemma}\label{lem:inC}
	Let $w$ be a synchronizing word with minimal switch count. If $vb$ is a prefix of $w$ such that $\mu(Svb)=\mu(Sv)+1$, then $Sv\subseteq C$ or $Sv\subseteq -C$.
\end{lemma}
\begin{proof} By Lemma \ref{lem:measurediff}, we find $n\in Sv$ (or $-n\in Sv$). By Lemma \ref{lem:prefix}, the word $u = b(ab)^{\frac{n}{3}-1}$ is a prefix of $w$. This word satisfies $Su \subseteq -C$ and furthermore $\mu(S\tilde u)=1$ for every prefix $\tilde u$ of $u$. This implies that $u$ is a prefix of $v$ and now Lemma \ref{lem:subsetC} gives $Sv \subseteq C$ or $Sv \subseteq -C$.
\end{proof}

Having collected these results, we are ready to provide the proof of Theorem \ref{theorem:lowerbound} in a few propositions. First we prove that the expression given in Theorem \ref{theorem:lowerbound} is a lower bound, by breaking up a synchronizing word in subwords that increase the measure in steps of size 1.

\begin{proposition}\label{prop:lower}
	Let $w$ be a synchronizing word for $\mathcal{A}_n$. Then \[|w|_s\geq \left\lceil \frac23 n(n-2)-1\right\rceil.\]
\end{proposition}
\begin{proof} The proof is given for $n$ divisible by 6. Let $w$ be a synchronizing word for $\mathcal{A}_n$, with minimal switch count. Assume that $w$ is the shortest such word. Then $\mu(S)=1$ and $|Sw|=1$ in $\mathcal{B}_n$, meaning that $\mu(Sw)=\frac23 n$. For $1\leq k\leq \frac23 n$, we let $u_k$ be the longest prefix  of $w$ such that $\mu(Su_k) = k$. By Lemma \ref{lem:measurediff}, such a prefix exists and if $2\leq k\leq\frac23 n$, then $|u_{k-1}|<|u_k|$. So for $2\leq k\leq\frac23 n$ we can write $u_k=u_{k-1}v_k$ for some non-empty word $v_k$. This implies that $w$ is the concatenation $v_1v_2\ldots v_{\frac23 n}$, where we set $v_1=u_1$. Furthermore, since $w$ is a synchronizing word with minimal switch count, we have $|v_{\frac23 n}|_s=1$.
	
	For all $2\leq k\leq\frac23 n$, if the first letter of $v_k$ would be $a$, then $u_{k-1}a$ would be a prefix of $w$ and $\mu(Su_{k-1}a) = \mu(Su_{k-1}) = k-1$ by Lemma \ref{lem:measurediff}. This contradicts the fact that $u_{k-1}$ is the longest prefix with this property, therefore the first letter of $v_k$ is $b$. Similarly, for $1\leq k\leq \frac23 n-1$, Lemma  \ref{lem:measurediff} implies that $n\in Su_{k}$, because otherwise $\mu(Su_kb) =\mu(Su_k)= k$. Since there are no states $q$ for which $qa=n$, we conclude that the last letter of $u_k$ is $b$. Hence also $v_k$ ends with $b$.
	
	Knowing the first and last letters of the words $v_k$, we can express the switch count of $w$ by
	\[
	|w|_s = 1+\sum_{k=1}^{\frac23 n} \left(|v_k|_s-1\right) = -\frac23 n+2+\sum_{k=1}^{\frac23 n-1} |v_k|_s.
	\]
	For $k=1$, we apply Corollary \ref{cor:prefix} to conclude that $|v_k|_s\geq \frac23n-1$, giving \[|w|_s = 1+\sum_{k=2}^{\frac23 n-1}|v_k|_s.\]
	
	For $2\leq k\leq \frac23 n-1$, we have that $u_kb$ is a prefix of $w$ and by definition $\mu(Su_kb)\neq k$. If $\mu(Su_kb)$ would be less than $k$, there would be a prefix $\tilde w$, longer than $u_k$ for which $\mu(S\tilde w)=k$, since the measure can only increase in steps of size 1 by Lemma \ref{lem:measurediff}. Therefore $\mu(Su_kb)=k+1$. Letting  $A_k = Su_{k-1}$, we observe that $\mu(A_k)=k-1$, $\mu(A_kb)=k$ and $\mu(A_kv_kb)=k+1$. Lemma \ref{lem:inC} gives $A_k\subseteq C$ (or $-C$). By Lemma \ref{lem:setpair}, there exist $p,q\in A_k$ such that $d(p,q)\leq k-1$ and $d(pv_kb,qv_kb) = k+1$. Lemma \ref{lem:pairdist} gives
	\[
	|v_kb|_s \geq \left\{
	\begin{array}{ll}
	\frac23 n+2k-1\qquad & 2\leq k\leq \frac{n}{3},\\
	2n-2k+1\qquad & \frac{n}{3}+1\leq k\leq \frac23 n-1.
	\end{array}
	\right.
	\]
	Since $v_k$ ends with $b$, we have $|v_k|_s=|v_kb|_s$, so that finally
	\[
	|w|_s\geq 1+\sum_{k=2}^{n/3} \left(\frac23 n+2k-1\right)+\sum_{k=n/3+1}^{2n/3-1} (2n-2k+1) = \frac23 n^2-\frac43 n-1 = \frac23 n(n-2)-1.
	\]
\end{proof}

In the previous proof we have seen that the subwords $v_k$ are at least as long as those given in Lemma \ref{lem:pairdist} and Lemma \ref{lem:prefix}. In the next result, we will see that we can actually construct a synchronizing word by choosing $v_kb$ equal to the words of these lemmata and taking $w=\prod_{k}v_k$. This gives that the switch count of $\mathcal{A}_n$ is exactly equal to the bound of Proposition \ref{prop:lower}.

\begin{proposition} Let $w$ be a synchronizing word for $\mathcal{A}_n$ with minimal switch count. Assume that $w$ is the shortest such word. Then $w$ is uniquely determined and $|w|_s = \left\lceil\frac23 n(n-2)-1\right\rceil$.
\end{proposition}

\begin{proof}
	Define the words $v_k$ for $1\leq k\leq \frac23 n$ as follows:
	\[
	\begin{array}{lll}
	v_1 &= b(ab)^{\frac{n}{3}-1},&\\
	v_kb &= b(ba)^k(ab)^{\frac{n}{3}}b\qquad &2\leq k\leq \frac{n}{3}-1,\\
	v_kb &= b(ba)^{\frac{2}{3}n-1}b^2&k=\frac{n}{3},\\
	v_kb & = (ba)^{n-k}b^2&\frac{n}{3}+1\leq k\leq \frac23 n-1,\\
	v_{\frac23 n}& = b. &
	\end{array}
	\]
	%
	Observe that $Sv_1 = -C$. Define intervals $I_k$ as follows
	\[
	I_k = \left\{\begin{array}{ll}
	\left[-\frac{n}{3}+1,n-2k\right] & 0\leq k\leq \frac{n}{2}-1, \\
	\left[-\frac{n}{3}+1,n-2k-1\right]\qquad & \frac{n}{2}\leq k\leq \frac23 n-1.
	\end{array}
	\right.
	\]
	Then $I_0 = C = I_1\cup \left\{n\right\}$. A straightforward check reveals that
	\[
	\left(I_{k-1}\cup\left\{n\right\}\right)v_k = -(I_k\cup\left\{n\right\})
	\]
	for $2\leq k\leq \frac23 n-1$. Also note that for $k = \frac23 n-1$, we have $I_k\cup\left\{n\right\} = \left\{-\frac{n}{3}+1,n\right\}$. This pair is synchronized by applying $v_{\frac23 n}=b$. We conclude that the concatenation
	\[
	w = \prod_{k=1}^{\frac23 n}v_k = b(ab)^{\frac{n}{3}-1}\prod_{k=2}^{\frac{n}{3}-1}\Bigl(b(ba)^k(ab)^{\frac n 3}\Bigr) b(ba)^{\frac23 n-1}b\prod_{k=\frac n 3+1}^{\frac23 n-1}\Bigl((ba)^{n-k}b\Bigr)b
	\]
	is a synchronizing word for $\mathcal{A}_n$ with switch count $\frac23 n(n-2)-1$. Uniqueness follows from the fact that the words constructed in Lemma \ref{lem:pairdist} and Lemma \ref{lem:prefix} are unique.
\end{proof}

\section{Cyclic automata}
\label{seccyc}

An automaton is called {\em cyclic} if one of the symbols $a$ is cyclic, that is, the $n$ states can be numbered by $0,1,\ldots,n-1$ in such a way that $qa = (q+1 \mod n)$ for all states $q$.
For cyclic automata it was proved by Dubuc \cite{dubuc} that \v{C}ern\'y's conjecture holds. In 1978 this was already proved by Pin in \cite{pin78} for $n$ being a prime number.
Note that the \v{C}ern\'y automaton ${\cal C}_n$ is cyclic for every $n$; it has linear switch count, more precisely, its switch count is exactly $2n-3$, as we observed in the introduction.

In this section we prove that for $n$ being a prime number, every synchronizing cyclic automaton on $n$ states has switch count at most $2n-3$. Part of our proof is similar to the proof in \cite{pin78} of \v{C}ern\'y's conjecture for cyclic automata for which $n$ is a prime number. For $n$ being non-prime, we expect that the switch count is still linear, but we give examples that it may exceed $2n-3$.

The result for $n$ being prime is based on the following lemma; its proof needs some algebra. We write $\mathbb{Z}$ for the set of integers.

\begin{lemma}
	\label{lem1}
	Let $n$ be a prime number and $Q = \{0,1,\ldots,n-1\}$.
	
	Let $A \in \mathbb{Z}$, $f,g : Q \to \mathbb{Z}$ satisfy
	\[ \sum_{i=0}^{n-1} f(i) g(i+k) \; = \; A\]
	for all $k = 0,1,\ldots,n-1$, in which the arguments of $g$ are taken modulo $n$. Then at least one of the functions $f$ and $g$ is constant.
\end{lemma}
\begin{proof}
	The cyclotomic polynomial $\Phi_n$ is defined by $\Phi_n[X] = X^{n-1}+X^{n-2}+\cdots+X+1$, note that $\Phi_n[X](X-1) = X^n - 1$.
	The complex number $\zeta = e^{2 \pi i/n}$ satisfies $\zeta \neq 1$ and $\zeta^n = 1$. So $\Phi_n[\zeta] = \frac{\zeta^n - 1}{\zeta -1} = 0$.
	A well-known fact from algebra is that $\Phi_n$ is irreducible, see e.g. \cite{W71}, Chapter 8, Section 60. A direct consequence is that no non-zero polynomial $P$ exists with integer coefficients and degree $< n-1$
	satisfying $P[\zeta] = 0$. If we would have
	$\sum_{j=0}^{n-1} f(j) \zeta^j = 0$ for $f : Q \to \mathbb{Z}$ being a non-constant function, then $(\sum_{j=0}^{n-1} f(j)X^j) - f(n-1)\Phi_n[X]$ would be such a non-existing polynomial $P$. So for $f : Q \to \mathbb{Z}$
	being non-constant we have $\sum_{j=0}^{n-1} f(j) \zeta^j \neq 0$.  Since $\zeta^n = 1$, in computing $\zeta^k$ we may take $k$ modulo $n$.
	
	Assume that neither $f$ nor $g$ is constant. Then according to the above fact we obtain $F = \sum_{j=0}^{n-1} f(j) \zeta^j \neq 0$ and $G = \sum_{j=0}^{n-1} g(j) \zeta^j \neq 0$.
	Since the conjugate of $\zeta$ is $\overline{\zeta} = \zeta^{-1}$, we also obtain  $\overline{F} = \sum_{j=0}^{n-1} f(j) \zeta^{-j} \neq 0$.
	We obtain
	\[\begin{array}{rcl}
	0 & = & \sum_{k=0}^{n-1} A \zeta^k \\
	& = & \sum_{k=0}^{n-1} \sum_{j=0}^{n-1} f(j) g(j+k) \zeta^k \\
	& = & \sum_{j=0}^{n-1} \sum_{k=0}^{n-1} f(j) g(j+k) \zeta^k \\
	& = & \sum_{j=0}^{n-1} f(j) \sum_{k=0}^{n-1} g(j+k) \zeta^k \\
	& = & \sum_{j=0}^{n-1} f(j) \sum_{k=0}^{n-1} g(k) \zeta^{k-j} \\
	& = & \sum_{j=0}^{n-1} f(j) \sum_{k=0}^{n-1} g(k) \zeta^k \zeta^{-j}\\
	& = & \sum_{j=0}^{n-1} f(j) \zeta^{-j} \sum_{k=0}^{n-1} g(k) \zeta^k \\
	& = & \overline{F} G,\end{array}\]
	contradicting $\overline{F} \neq 0 \neq G$.
\end{proof}

If $n$ is not prime then Lemma \ref{lem1} does not hold, as is shown by the following example. Write $n=pq$ for $p,q > 1$. Define $f,g : \{0,1,2,\ldots,n-1\} \to \mathbb{Z}$ by $f(i)=1$ for
$i < p$ and $f(i)=0$ for $i \geq p$, and $g(i) = 1$ if $i$ is divisible by $p$ and  $g(i) = 0$ otherwise. Then it is easily checked that
$\sum_{i=0}^{n-1} f(i) g(i+k) = 1$ for all $k$, but neither $f$ nor $g$ is a constant function.

\begin{theorem}
	\label{thmcyc}
	Let $n$ be a prime number. Then the maximal switch count of any synchronizing cyclic automaton on $n$ states is $2n-3$.
\end{theorem}
\begin{proof}
	We already observed that the switch count $2n-3$ is achieved for the \v{C}ern\'y automaton ${\cal C}_n$; it remains to prove that the switch count cannot exceed $2n-3$.
	
	Let $a$ be the cyclic symbol. Let  $Q = \{0,1,\ldots,n-1\}$ be the set of states, numbered in such a way that $qa = (q+1 \mod n)$ for all states $q$. Since the automaton is synchronizing it contains a non-injective symbol $b$. We will prove that there is a synchronizing word of the shape $(b a^*)^k b$ for $k \leq n-2$, having switch count $\leq 2n-3$, so proving the theorem.
	
	The reasoning in the power automaton is backward, that is, we start by $V_0$ for which $V_0 b$ is a singleton, and then construct $V_1,V_2,\ldots,V_k$ satisfying $V_i b a^m \subseteq V_{i-1}$ for some $m$, for all $i=1,2,\ldots,k$, and $V_k = Q$: this yields a synchronizing word. Since $b$ is non-injective, $V_0$ can be chosen satisfying $|V_0| \geq 2$. The construction repeats applying the following claim.
	\begin{quote}
		{\bf Claim:} Let $V \subseteq Q$ satisfy $\emptyset \neq V \neq Q$. Then a set $W \subseteq Q$ and a number $m$ exist such that $W b a^m \subseteq V$ and
		$|W| > |V|$.
	\end{quote}
	Applying this claim on $V = V_i$ for $i \geq 0$ we define $V_{i+1} = W$, since $|V_0| \geq 2$ and $|V_{i+1}| > |V_i|$ for all $i$, we obtain $V_k = Q$ for $k \leq n-2$, yielding the required synchronizing word of the shape $(b a^*)^k b$, proving the theorem. So it remains to prove the claim.
	
	Assume the claim does not hold. For any set $A \subseteq Q$ write $A-m = \{q \in Q \mid (q+m \mod n) \in A\}$.
	Due to the numbering we have
	$(A-m) a^m = A$ for all $m$. Since the claim does not hold and $(V-m)b^{-1} b a^m \subseteq V$, we have $| (V-m)b^{-1} | \leq |V|$ for all $m$. Combining this with
	\[ \sum_{i=0}^{n-1} | (V-i)b^{-1} | = \sum_{i=0}^{n-1} \sum_{q \in V-i} | q b^{-1} | = |V| \sum_{q \in Q} | q b^{-1} | = n|V|\]
	yields $| (V-m)b^{-1} | = |V|$ for all $m$. For $i \in Q$ define $f(i) = |i b^{-1}|$ and $g(i) = 1$ if $i \in V$ and $g(i) = 0$ if $i \not\in V$. Writing the arguments of $g$ modulo $n$, then for all $m$ we have
	\[ |V| = | (V-m)b^{-1} | = \sum_{q \in V-m} | q b^{-1} | = \sum_{i=0}^{n-1} g(i+m) | i b^{-1} | = \sum_{i=0}^{n-1} g(i+m) f(i).\]
	So by lemma \ref{lem1} we obtain that  at least one of the functions $f$ and $g$ is constant. However, $f$ is not constant since $b$ is not injective and $g$ is not constant since $\emptyset \neq V \neq Q$, contradiction. This proves the claim and concludes the proof of the theorem.
\end{proof}

For Theorem \ref{thmcyc} the requirement of $n$ being prime is essential, as is shown by the following example.

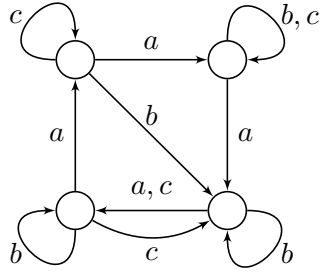
\begin{center}
	\begin{tikzpicture}[-latex',node distance =2 cm and 2cm ,on grid ,
	semithick , state/.style ={ circle ,top color =white , bottom
		color = white!20 , draw, black , text=black , minimum width =.5
		cm}]
	
	\node[state] (1) {};
	\node[state] (4) [below =of 1] {};
	\node[state] (2) [right =of 1] {};
	\node[state] (3) [right =of 4] {};
	\path (1) edge node[above] {$a$} (2);
	\path (1) edge node[above] {$b$} (3);
	\path (2) edge node[right] {$a$} (3);
	\path (3) edge node[above] {$a,c$} (4);
	\path (4) edge [looseness=1, in = 210, out = 330] node[below] {$c$} (3);
	\path (4) edge node[left] {$a$} (1);
	\path (2) edge [loop right,looseness=8, in = 0, out = 90] node[right] {$b,c$} (2);
	\path (3) edge [loop right,looseness=8, in = 270, out = 0] node[right] {$b$} (3);
	\path (4) edge [loop right,looseness=8, in = 180, out = 270] node[left] {$b$} (4);
	\path (1) edge [loop right,looseness=8, in = 90, out = 180] node[left] {$c$} (1);
	\end{tikzpicture}
\end{center}

In this automaton all three symbols are needed for synchronization, in contrast to the proof of Theorem \ref{thmcyc} where only the cyclic symbol $a$ and the non-injective symbol $b$ suffice.
The switch count is 6, being strictly larger than $2n-3 = 5$. A corresponding synchronizing word is $babacb$.

\section{Conclusions and open problems}
\label{secconcl}

We investigated synchronizing automata with high switch count rather than high synchronization length. For synchronization length for all $n$ the highest known value $(n-1)^2$ is achieved by the binary automaton ${\cal C}_n$,
and for $n > 6$ it is conjectured that ${\cal C}_n$ is the only one for which this value is achieved. For switch count the situation is quite different: for all $n \leq 9$ the best binary automata are outperformed by non-binary
automata, and for all $n \leq 11$ the automata with the best known value are not unique. For $n \geq 12$ the best known values $\left\lceil \frac{2}{3}n(n-2)-1\right\rceil$ are obtained by the binary automata ${\cal A}_n$. In particular, this establishes the lower bound $\frac23 n^2 +O(n)$ for the maximal synchronization length for the class of power closed automata.

Open problems include
\begin{itemize}
	\item For $n \geq 10$, does a binary automaton on $n$ states exist with switch count exceeding $\left\lceil \frac{2}{3}n(n-2)-1\right\rceil $?
	\item For $n \geq 12$, does an automaton on $n$ states exist, different from ${\cal A}_n$, with switch count $\geq  \left\lceil \frac{2}{3}n(n-2)-1\right\rceil $?
	\item Does $\limsup_{n\rightarrow\infty} \frac{\SW(n)}{n^2} = \frac{2}{3}$ hold?
	\item Is the switch count of cyclic automata linear for all $n$? (We only showed this for $n$ being prime.)
\end{itemize}

\bibliographystyle{plain}
\bibliography{ref}

\end{document}